\documentclass[12pt]{article}

\usepackage{pstricks}
\usepackage{amsfonts}
\usepackage{mathrsfs}
\usepackage{pstricks}
\usepackage[all]{xy}
\usepackage{amscd}
\usepackage{algorithmic,algorithm}
\usepackage{amsmath,amssymb,amsthm,graphicx,titlesec}
\newtheorem{theorem}{Theorem}[section]
\newtheorem{lemma}[theorem]{Lemma}

\newtheorem{corollary}[theorem]{Corollary}
\newtheorem{definition}{Definition }[section]
\newtheorem{example}{Example }[section]

\title{The Dimension of Spline Spaces with Highest Order Smoothness over Hierarchical T-meshes}
\author{Meng Wu \quad Jiansong Deng\footnote{Corresponding author, email: dengjs@ustc.edu.cn}
\quad Falai Chen \\
 School of Mathematical Sciences \\
 University of Science and Technology of China\\
 Hefei, Anhui, 230026 \\
 P. R. China}
\date{}

\begin{document}
\maketitle
\begin{abstract}
This paper discusses the dimension of spline spaces
$\mathbf{S}(m,n,m-1,n-1,\mathscr{T})$ over certain type of
hierarchical T-meshes. The major step is to set up a bijection
between the spline space $\mathbf{S}(m,n,m-1,n-1,\mathscr{T})$ and a
univariate spline space whose definition depends on the l-edges of
the extended T-mesh. We decompose the univariate spline space into
direct sums in the sense of isomorphism using the theory of the
short exact sequence in homological algebra. According to the
decomposition of the univariate spline space, the dimension formula
of the spline space $\mathbf{S}(m,n,m-1,n-1,\mathscr{T})$ over
certain type of hierarchical T-mesh is presented. A set of basis functions
of the spline space is also constructed.

\medskip


\end{abstract}

\section{Introduction}
Non-Uniform Rational B-Splines (NURBS) are popular tool to
represent surface models in Computer Aided Geometric Design (CAGD)
and Computer Graphics. However, due to the tensor-product structure of
NURBS, local refinement of surface models based on NURBS is impossible,
and NURBS models generally contain large number of superfluous control points.
To overcome the above drawbacks, Sederberg et al. introduced T-splines, the control
meshes of which allow T-junctions (\cite{Sederberg2003,Sederberg2004}).
T-splines provide local refinement strategy and can reduce the large number of superfluous
control points in NURBS models.



In \cite{Deng2006}, the concept of spline spaces over T-meshes is
introduced.  Different from T-splines, a spline over a T-mesh is
a single polynomial within each cell of the T-mesh, and it achieves
the specified smoothness across the common edges. Spline spaces over
T-meshes are suitable for geometric modeling\cite{Deng2006PHT} and
they can be applied for analysis naturally, since it is easy to
do standard Finite Element Method analysis based on splines over T-meshes
\cite{FEM}.

One major issue in the theory of splines over T-meshes is to study
the dimension of the spline spaces. There have been a few literature
focusing on the problem so far. In \cite{Deng2006}, the dimension
formula for the spline space
$\mathbf{S}(m,n,\alpha,\beta,\mathscr{T})$ is obtained with
constraints $m\ge 2\alpha+1,n\ge 2\beta+1$. The result is further
improved by Li et al. \cite{LiWang}. For the spline spaces with
highest order smoothness where $m \le 2\alpha, n\le 2\beta$, the
authors of the current paper derived a dimension formula for $C^1$
biquadratic spline spaces (that is,
$\mathbf{S}(2,2,1,1,\mathscr{T})$) over hierarchical
T-meshes\cite{Deng2008:21}. Recently, B. Mourrain gives a general
formula for the spline spaces
$\mathbf{S}(m,n,\alpha,\beta,\mathscr{T})$ by homological
techniques\cite{Mourrain2010}. Unfortunately there is a term in the
dimension formula which is very hard to compute in practice.

Contrary to the above positive results, Li and Chen have showed that
the dimension of spline spaces with highest order smoothness over
T-meshes may depend on the geometry besides the topology information
of the T-meshes \cite{InstabilityDimLi}. The result suggests that it
is vain to study the dimension formula over general T-meshes in the
case of highest order of smoothness. In this paper, the dimension
formula of spline spaces $\mathbf{S}(m,n,m-1,n-1,\mathscr{T})$ over
certain type of hierarchical T-meshes will be explored.

There are several methods for establishing the dimension of spline
spaces, such as the B-net method~\cite{Deng2006}, and the smoothing
cofactor method~\cite{WRH1979} and homology
method~\cite{Billera1988}. In the smoothing cofactor method, the
cofactor of a spline associated with the common edge of two adjacent
cells is a univariate polynomial. This is similar to homology theory
in topology. Therefore, using the smoothing cofactor method and
according to a result in \cite{Deng2008:21}, we construct an
isomorphism between a spline space with highest order smoothness
over a T-mesh and a univariate spline space satisfying some
conditions. For a certain type of hierarchical T-mesh, by giving an
order of the interior l-edges of the extended T-mesh, we can then
decompose the univariate spline space into a direct sum, and thus
give a dimension formula for the spline space over hierarchial
T-meshes.

The rest of the paper is organized as follows. In Section 2,
the definitions and some results regarding  T-meshes and spline spaces over
T-meshes are reviewed. In Section 3, an equivalence is set up between
the spline over a T-mesh and a univariate spline space.
The proof of the dimension formula is presented in Section 4. Some examples
are also provided. In Section 5, we conclude the paper with
future research problems. We leave the proof for a key lemma and the construction
of a set of basis functions of the spline space in the appendix.

\section{Spline Spaces over T-meshes}
In this section, we first review some concepts about T-meshes and
spline spaces over T-meshes.

\subsection{Spline spaces over T-meshes}

A {\bf T-mesh} is a rectangular grid that allows T-junctions.
For simplicity, in this paper we consider only regular T-meshes whose boundary grid lines form a rectangle. We adopt
the same definitions for \textbf{vertex}, \textbf{edge}, and \textbf{cell} as in
 \cite{Deng2006}, and the definitions for \textbf{l-edge}, \textbf{interior
 l-edge}, \textbf{associated tensor product mesh} are borrowed from
~\cite{Deng2008:21}. A grid point in a T-mesh is called a
 vertex of the T-mesh. Vertices of a T-mesh are divided into different
 types. For example, in Figure \ref{fig:concerpt}, $\{b_i\}^{i=10}_{i=1}$ are boundary vertices and
 $\{v_i\}^{i=5}_{i=1}$ are interior vertices. $v_2$ is a crossing
 vertex and $\{b_i\}^{i=10}_{i=1}\cup\{v_i\}^{i=5}_{i=1}-\{v_2,b_1,b_3,b_6,b_8\}$
 are T-vertices. The line segment connecting two adjacent vertices
 on a grid line is called an edge of T-mesh such as $v_4v_5,b_9b_{10},v_2v_3$
 in Figure \ref{fig:concerpt}. $b_2v_3$ is a large edge (l-edge for short), which is the longest possible line segment
 consisting of several edges. The boundary of a regular T-mesh consists of four l-edges,which are called boundary l-edges.
 The other l-edges in T-mesh are called interior l-edges.
 A regular T-mesh can be extended to a tensor product mesh,
 called the associated tensor-product mesh, by
 extending all the interior large edges to the boundary, see Figure \ref{fig:concerpt}.

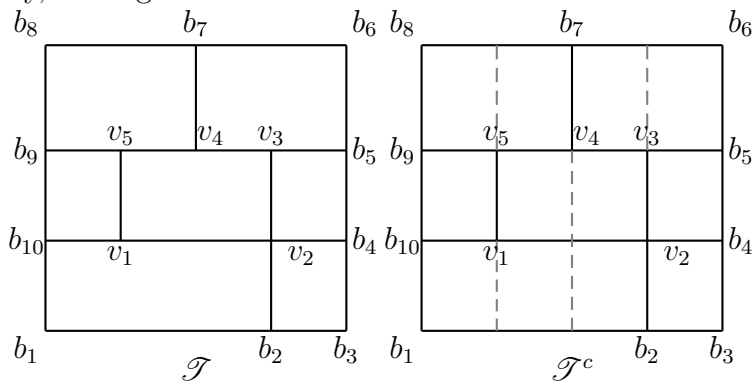
\begin{figure}[!htb]
   \begin{center}
    \begin{pspicture}(0.5,0.5)(10.5,5)
    \psline[linewidth=0.8pt](1,1)(1,4.8)
    \psline[linewidth=0.8pt](2,2.2)(2,3.4)
    \psline[linewidth=0.8pt](3,3.4)(3,4.8)
    \psline[linewidth=0.8pt](4,1)(4,3.4)
    \psline[linewidth=0.8pt](5,1)(5,4.8)
    \psline[linewidth=0.8pt](1,1)(5,1)
    \psline[linewidth=0.8pt](1,2.2)(5,2.2)
    \psline[linewidth=0.8pt](1,3.4)(5,3.4)
    \psline[linewidth=0.8pt](1,4.8)(5,4.8)
    \rput(0.75,0.75){$b_1$}
    \rput(4,0.75){$b_2$}
    \rput(5,0.75){$b_3$}
    \rput(5.25,2.2){$b_4$}
    \rput(5.25,3.4){$b_5$}
    \rput(5.25,5.1){$b_6$}
    \rput(3,5.1){$b_7$}
    \rput(0.75,5.1){$b_8$}
    \rput(0.75,3.4){$b_9$}
    \rput(0.75,2.2){$b_{10}$}
    \rput(2,2){$v_1$}
    \rput(4.4,2){$v_2$}
    \rput(4,3.6){$v_3$}
    \rput(3.2,3.6){$v_4$}
    \rput(2,3.6){$v_5$}
    \rput(3,0.5){$\mathscr{T}$}
    \psline[linewidth=0.8pt](6,1)(6,4.8)
    \psline[linewidth=0.8pt](7,2.2)(7,3.4)
    \psline[linewidth=0.8pt](8,3.4)(8,4.8)
    \psline[linewidth=0.8pt](9,1)(9,3.4)
    \psline[linewidth=0.8pt](10,1)(10,4.8)
    \psline[linewidth=0.8pt](6,1)(10,1)
    \psline[linewidth=0.8pt](6,2.2)(10,2.2)
    \psline[linewidth=0.8pt](6,3.4)(10,3.4)
    \psline[linewidth=0.8pt](6,4.8)(10,4.8)
    \psline[linewidth=0.8pt,linestyle=dashed,linecolor=gray](7,1)(7,2.2)
    \psline[linewidth=0.8pt,linestyle=dashed,linecolor=gray](7,3.4)(7,4.8)
    \psline[linewidth=0.8pt,linestyle=dashed,linecolor=gray](8,1)(8,3.4)
    \psline[linewidth=0.8pt,linestyle=dashed,linecolor=gray](9,3.4)(9,4.8)
    \rput(5.75,0.75){$b_1$}
    \rput(9,0.75){$b_2$}
    \rput(10,0.75){$b_3$}
    \rput(10.25,2.2){$b_4$}
    \rput(10.25,3.4){$b_5$}
    \rput(10.25,5.1){$b_6$}
    \rput(8,5.1){$b_7$}
    \rput(5.75,5.1){$b_8$}
    \rput(5.75,3.4){$b_9$}
    \rput(5.75,2.2){$b_{10}$}
    \rput(7,2){$v_1$}
    \rput(9.4,2){$v_2$}
    \rput(9,3.6){$v_3$}
    \rput(8.2,3.6){$v_4$}
    \rput(7,3.6){$v_5$}
    \rput(8,0.5){$\mathscr{T}^c$}
    \end{pspicture}
  \caption{A T-mesh $\mathscr{T}$ and its associated tensor product mesh $\mathscr{T}^c$.\label{fig:concerpt}}
 \end{center}
 \end{figure}

Given a T-mesh $\mathscr{T}$, $\mathscr{F}$ is the set of all the
cells of $\mathscr{T}$ and $\Omega$ is the region occupied by cells
in $\mathscr{F}$. Spline spaces over T-meshes are defined by
\begin{align}
\label{Splinedef}
  \mathbf{S}(m,n,\alpha,\beta,\mathscr{T}):&=\{f(x,y)\in
    C^{\alpha,\beta}(\Omega):f(x,y)|\phi\in\mathbb{P}_{mn},\forall\phi\in\mathscr{F}\},
\end{align}
where $\mathbb{P}_{mn}$ is the space of all the polynomials with
bi-degree $(m,n)$, and $C^{\alpha,\beta}$ is the space consisting of
all the bivariate functions that are continuous in $\Omega$ with
order $\alpha$ along the $x$ direction and order $\beta$ along the $y$
direction. In this paper, we will focus on the spline space
$\mathbf{S}(m,n,m-1,n-1,\mathscr{T})$.

\subsection{Hierarchical T-meshes and extended T-meshes}


 A hierarchical T-mesh is a special type of T-mesh that has a natural
 level structure~\cite{Deng2006PHT}. It is defined in a recursive fashion. Initially a tensor product mesh (level 0) is presumed. From level $k$ to $k+1$,
 a cell is subdivided at level $k$ into four sub-cells, which are cells at
 level $k+1$, by connecting the middle points of the opposite edges with two straight lines.
 Figure~\ref{fig:HTmesh} illustrates a sequence of hierarchical T-meshes.

\begin{figure}[htpb]
 \begin{center}
\begin{pspicture}(0,0.5)(13,5.5)
\psline[linewidth=0.8pt](1,1)(4,1)
\psline[linewidth=0.8pt](1,2)(4,2)
\psline[linewidth=0.8pt](1,3.5)(4,3.5)
\psline[linewidth=0.8pt](1,4.5)(4,4.5)
\psline[linewidth=0.8pt](1,5.5)(4,5.5)

\psline[linewidth=0.8pt](1,1)(1,5.5)
\psline[linewidth=0.8pt](2,1)(2,5.5)
\psline[linewidth=0.8pt](3,1)(3,5.5)
\psline[linewidth=0.8pt](4,1)(4,5.5)

\rput(2.5,0.5){Level 0}

\psline[linewidth=0.8pt](5,1)(8,1)
\psline[linewidth=0.8pt](5,2)(8,2)
\psline[linewidth=0.8pt](5,3.5)(8,3.5)
\psline[linewidth=0.8pt](5,4.5)(8,4.5)
\psline[linewidth=0.8pt](5,5.5)(8,5.5)

\psline[linewidth=0.8pt](5,1)(5,5.5)
\psline[linewidth=0.8pt](6,1)(6,5.5)
\psline[linewidth=0.8pt](7,1)(7,5.5)
\psline[linewidth=0.8pt](8,1)(8,5.5)

\psline[linewidth=0.8pt,linecolor=red](5,1.5)(7,1.5)
\psline[linewidth=0.8pt,linecolor=red](6,2.75)(7,2.75)
\psline[linewidth=0.8pt,linecolor=red](5,4)(8,4)

\psline[linewidth=0.8pt,linecolor=red](5.5,1)(5.5,2)
\psline[linewidth=0.8pt,linecolor=red](5.5,3.5)(5.5,4.5)
\psline[linewidth=0.8pt,linecolor=red](6.5,1)(6.5,4.5)
\psline[linewidth=0.8pt,linecolor=red](7.5,3.5)(7.5,4.5)
\rput(6.5,0.5){Level 1}

\psline[linewidth=0.8pt](9,1)(12,1)
\psline[linewidth=0.8pt](9,2)(12,2)
\psline[linewidth=0.8pt](9,3.5)(12,3.5)
\psline[linewidth=0.8pt](9,4.5)(12,4.5)
\psline[linewidth=0.8pt](9,5.5)(12,5.5)

\psline[linewidth=0.8pt](9,1)(9,5.5)
\psline[linewidth=0.8pt](10,1)(10,5.5)
\psline[linewidth=0.8pt](11,1)(11,5.5)
\psline[linewidth=0.8pt](12,1)(12,5.5)

\psline[linewidth=0.8pt](9,1.5)(11,1.5)
\psline[linewidth=0.8pt](10,2.75)(11,2.75)
\psline[linewidth=0.8pt](9,4)(12,4)

\psline[linewidth=0.8pt](9.5,1)(9.5,2)
\psline[linewidth=0.8pt](9.5,3.5)(9.5,4.5)
\psline[linewidth=0.8pt](10.5,1)(10.5,4.5)
\psline[linewidth=0.8pt](11.5,3.5)(11.5,4.5)

\psline[linewidth=0.8pt,linecolor=red](9,1.25)(9.5,1.25)
\psline[linewidth=0.8pt,linecolor=red](10,1.25)(10.5,1.25)
\psline[linewidth=0.8pt,linecolor=red](10,1.75)(10.5,1.75)
\psline[linewidth=0.8pt,linecolor=red](10,2.375)(11,2.375)
\psline[linewidth=0.8pt,linecolor=red](10.5,3.75)(11.5,3.75)

\psline[linewidth=0.8pt,linecolor=red](9.25,1)(9.25,1.5)
\psline[linewidth=0.8pt,linecolor=red](10.25,1)(10.25,2.75)
\psline[linewidth=0.8pt,linecolor=red](10.75,2)(10.75,2.75)
\psline[linewidth=0.8pt,linecolor=red](10.75,3.5)(10.75,4)
\psline[linewidth=0.8pt,linecolor=red](11.25,3.5)(11.25,4)
\rput(10.5,0.5){Level 2}
\end{pspicture}
  \caption{Hierarchical T-meshes.\label{fig:HTmesh}}
\end{center}
\end{figure}
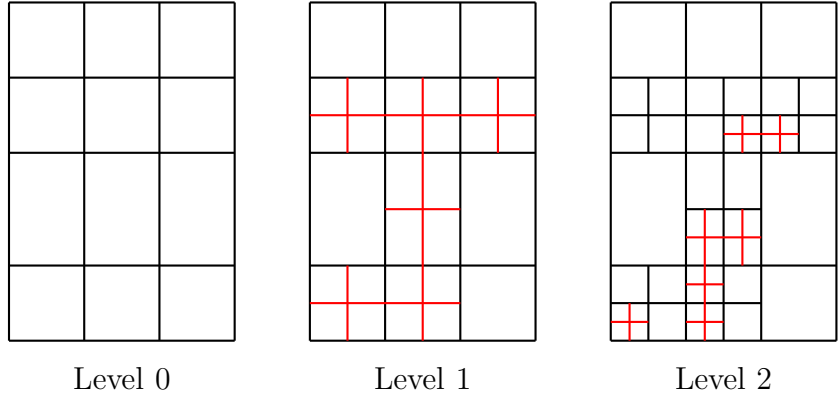

For a T-mesh $\mathscr{T}$,  the extended T-mesh
$\mathscr{T}^\varepsilon$ associated with $\mathscr{T}$ is an
enlarged T-mesh by copying each horizontal boundary line of
$\mathscr{T}$ $m$ times, and each vertical boundary line of
$\mathscr{T}$ $n$ times, and by extending all the line segments with
an end point on the boundary of $\mathscr{T}$~\cite{Deng2008:21}.
This can be made precise by the following example.
Figure~\ref{exTmesh} illustrates a T-mesh (left) and the extended
T-mesh (right) associated with degree $(3,3)$.

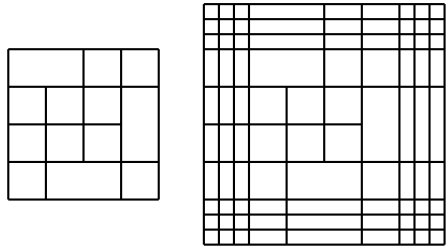
\begin{figure}[htpb]
\begin{center}
\begin{pspicture}(0,0)(6.4,3.2)
\psline[linewidth=0.8pt](0.6,0.6)(2.6,0.6)
\psline[linewidth=0.8pt](0.6,1.1)(2.6,1.1)
\psline[linewidth=0.8pt](0.6,1.6)(2.1,1.6)
\psline[linewidth=0.8pt](0.6,2.1)(2.6,2.1)
\psline[linewidth=0.8pt](0.6,2.6)(2.6,2.6)
\psline[linewidth=0.8pt](0.6,0.6)(0.6,2.6)
\psline[linewidth=0.8pt](1.1,0.6)(1.1,2.1)
\psline[linewidth=0.8pt](1.6,1.1)(1.6,2.6)
\psline[linewidth=0.8pt](2.1,0.6)(2.1,2.6)
\psline[linewidth=0.8pt](2.6,0.6)(2.6,2.6)
\psline[linewidth=0.8pt](3.2,0)(6.4,0)
\psline[linewidth=0.8pt](3.2,0.2)(6.4,0.2)
\psline[linewidth=0.8pt](3.2,0.4)(6.4,0.4)
\psline[linewidth=0.8pt](3.2,0.6)(6.4,0.6)
\psline[linewidth=0.8pt](3.2,1.1)(6.4,1.1)
 \psline[linewidth=0.8pt](3.2,1.6)(5.3,1.6)
\psline[linewidth=0.8pt](3.2,2.1)(6.4,2.1)
\psline[linewidth=0.8pt](3.2,2.6)(6.4,2.6)
\psline[linewidth=0.8pt](3.2,2.8)(6.4,2.8)
\psline[linewidth=0.8pt](3.2,3)(6.4,3)
\psline[linewidth=0.8pt](3.2,3.2)(6.4,3.2)
\psline[linewidth=0.8pt](3.2,0)(3.2,3.2)
\psline[linewidth=0.8pt](3.4,0)(3.4,3.2)
\psline[linewidth=0.8pt](3.6,0)(3.6,3.2)
\psline[linewidth=0.8pt](3.8,0)(3.8,3.2)
\psline[linewidth=0.8pt](4.3,0)(4.3,2.1)
\psline[linewidth=0.8pt](4.8,1.1)(4.8,3.2)
\psline[linewidth=0.8pt](5.3,0)(5.3,3.2)
\psline[linewidth=0.8pt](5.8,0)(5.8,3.2)
\psline[linewidth=0.8pt](6,0)(6,3.2)
\psline[linewidth=0.8pt](6.2,0)(6.2,3.2)
\psline[linewidth=0.8pt](6.4,0)(6.4,3.2)
\end{pspicture}
\caption{A T-mesh $\mathscr{T}$ and its extended T-mesh associated
with degree $(3,3)$. \label{exTmesh}}
\end{center}
\end{figure}

\subsection{Homogeneous boundary conditions}

A spline space over a given T-mesh $\mathscr{T}$ with homogeneous boundary conditions is defined by~\cite{Deng2008:21}
$$\overline{\mathbf{S}}(m,n,\alpha,\beta,\mathscr{T}):=\{f(x,y)\in
C^{\alpha,\beta}(\mathbb{R}^2):f(x,y)|_{\phi}\in\mathbb{P}_{mn},\forall\phi\in\mathscr{F}
\mbox{ and } f|_{\mathbb{R}^2 \backslash \Omega}\equiv 0\},$$
where $\mathbb{P}_{mn},\mathscr{F}$ is defined as before. One important observation
in \cite{Deng2008:21} is that the two spline spaces $\mathbf{S}(m,n,m-1,n-1,\mathscr{T})$ and
$\overline{\mathbf{S}}(m,n,m-1,n-1,\mathscr{T}^\varepsilon)$ are closely related.

\begin{theorem} \cite{Deng2008:21}
  Given a T-mesh $\mathscr{T}$, and let $\mathscr{T}^\varepsilon$ be the
  extended T-mesh associated with
  $\mathbf{S}(m,n,m-1,n-1,\mathscr{T})$. Then
  \begin{align}
    \mathbf{S}(m,n,m-1,n-1,\mathscr{T})&=\overline{\mathbf{S}}(m,n,m-1,n-1,\mathscr{T}^\varepsilon)|_{\mathscr{T}},\\
     \dim\mathbf{S}(m,n,m-1,n-1,\mathscr{T})&=\dim\overline{\mathbf{S}}(m,n,m-1,n-1,\mathscr{T}^\varepsilon).
    \end{align}
\end{theorem}

Based on the above theorem, we only have to consider spline spaces over T-meshes with
homogeneous boundary conditions.

\section{Equivalent spline spaces}
In this section, we will give an equivalent description of spline
spaces $\overline{\mathbf{S}}(m,n,m-1,n-1,\mathscr{T}^\varepsilon)$
by using the smoothing cofactor method.

\subsection{Smoothing cofactor method}

As before, let $\mathscr{T}$ and $\mathscr{T}^\varepsilon$ be a
T-mesh and its extension respectively. Let $\mathscr{F}=\{C_1,C_2,\cdots,C_k\}$ be
the set of all the cells in $\mathscr{T}^\varepsilon$, and
$\Omega_i$ be the region occupied by $C_i\in\mathscr{F}$. Denote
$\Omega_{k+1}=\mathbb{R}^2 \backslash \cup_{i=1}^{k}\Omega_i$ and
$\mathscr{F}^\varepsilon=\{\Omega_1,\Omega_2,\cdots,\Omega_{k+1}\}$.

Let $U_1,U_2,U_3,U_4~\in\mathscr{F}^\varepsilon$ be four regions in
adjacent positions as shown in Figure~\ref{cells}. For a spline
function
$f(x,y)\in\overline{\mathbf{S}}(m,n,m-1,n-1,\mathscr{T}^\varepsilon)$,
denote $f_i(x,y)$ be the bivariate polynomial which coincides with
$f(x,y)$ on $U_i$, $i=1,2,3,4$. Note that if $f_i(x,y)=f_j(x,y)$ for
two adjacent cells $U_i$ and $U_j$, then $U_i$ and $U_j$ can be
merged into a single region, and in this case $(x_0,y_0)$ is a
T-junction.

 \begin{figure}[htp]
 \begin{center}
 \begin{pspicture}(0,0)(4,4)
\psline[linewidth=0.8pt](0,2)(4,2)
\psline[linewidth=0.8pt](2,0)(2,4) \rput(1,1){\Large{$U_2$}}
\rput(3,1){\Large{$U_3$}} \rput(3,3){\Large{$U_4$}}
\rput(1,3){\Large{$U_1$}} \psdots(2,2)\rput(2.6,1.6){$(x_0,y_0)$}
 \end{pspicture}
  \caption{Smoothing conditions in adjacent cells}\label{cells}
 \end{center}
\end{figure}
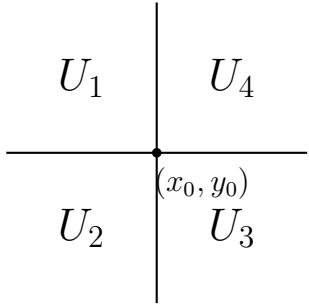

By the smoothing cofactor method, we have the following
relationship between $f_i(x,y)$:
\begin{lemma}\cite{LiWang,InstabilityDimLi}
\label{lemma1:fi relationship} Let $f_i(x,y),i=1,2,3,4$ be defined
as in the preceding paragraph. Then there exist a constant
$k\in\mathbb{R}$ and polynomials $a(y)\in\mathbb{P}_n[y]$,
$b(x)\in\mathbb{P}_m[x]$ such that
\begin{align*}
  f_1(x,y)&= f_2(x,y)+b(x)(y-y_0)^n ,\\
  f_3(x,y)&= f_2(x,y)+a(y)(x-x_0)^m ,\\
  f_4(x,y)&=
  f_2(x,y)+a(y)(x-x_0)^m+b(x)(y-y_0)^n+k(x-x_0)^m(y-y_0)^n,
\end{align*}
where $(x_0,y_0)$ is the vertex of $\mathscr{T}^ \varepsilon$ in
Figure~\ref{cells}. Furthermore, $a(y),b(x)$ and $k$ are uniquely
determined by $f_i(x,y)$, $i=1,2,3,4$, specifically
\begin{align}
\label{a(y)}
a(y)&=\frac{1}{m!}\left(\frac{\partial^mf_3(x_0,y)}{\partial
x^m}-\frac{\partial^mf_2(x_0,y)}{\partial x^m}\right),
\end{align}
\begin{align}
\label{b(x)}
b(x)&=\frac{1}{n!}\left(\frac{\partial^nf_1(x,y_0)}{\partial
y^n}-\frac{\partial^nf_2(x,y_0)}{\partial y^n}\right),
\end{align}
and
\begin{align}\label{kvalue}
    k&=\frac{1}{m!}\frac{1}{n!}\frac{\partial^{m+n}(f_2(x,y)+f_4(x,y)-f_1(x,y)-f_3(x,y))}{\partial x^m\partial
    y^n}.
\end{align}
$a(y)$ and $b(x)$ are smoothing co-factors associated the edges
between adjacent cells. The constant $k$ is called the conformality factor
associated with the common vertex $(x_0,y_0)$.

\end{lemma}

Based on the above lemma, we convert the dimension problem into the study of
a univariate spline space. We first introduce the following definition.

\begin{definition}
Let $E^h$ be a horizontal l-edge in $\mathscr{T}^\varepsilon$, and
$\{v_1,v_2,\cdots, v_r\}$ be vertices on $E^h$. Assume the $x$-coordinates of the
vertices are $x_1<x_2<\ldots<x_r$ respectively. 
We define a univariate spline space with homogeneous boundary conditions associated with $E^h$:
\begin{align}
  \overline{\mathbf{S}}(m,m-1,E^h): =& \left\{p(x)\in C^{m-1}(\mathbb{R}):p(x)|_{[x_i,x_{i+1}]}=p_i(x)\in\mathbb{P}_m[x],\right. \nonumber \\
  & i=1,2,\ldots,r-1, \left.\mbox{ and }~ p(x)|_{\mathbb{R}\backslash [x_1,x_r]}\equiv 0\right\}.
\end{align}
Similarly, for a vertical l-edge $E^v$, we can define a univariate spline space
$\overline{\mathbf{S}}(n,n-1,E^v)$ associated with it.
\end{definition}

By the homogeneous boundary conditions, it is easy to see that for any $p(x)\in\overline{\mathbf{S}}(m,m-1,E^h)$, there exist constants $k_1,\ldots,k_r$ such that
\begin{equation}\label{equation1 of l-edge}
\sum_{i=1}^rk_i(x-x_i)^m \equiv 0.
\end{equation}
Equation \eqref{equation1 of l-edge} is equivalent to a linear
system associated with $E^h$:
\begin{align}\label{linear equation of eq1}
\left\{
  \begin{array}{ll}
    \sum_{i=1}^rk_i=0, \\[2mm]
    \sum_{i=1}^rk_ix_i = 0,  \\[2mm]
    \cdots,  \\[2mm]
    \sum_{i=1}^rk_ix_i^m = 0.
  \end{array}
\right.
\end{align}

 For a vertical l-edge $E^v$, there is a similar equation
\begin{equation}\label{equation2 of l-edge}
\sum_{i=1}^sk_i(y-y_i)^n \equiv 0,
\end{equation}
which is also equivalent to a linear system associated with $E^v$:

\begin{align}\label{linear equation of eq2}
\left\{
  \begin{array}{ll}
    \sum_{i=1}^sk_i=0, \\[2mm]
    \sum_{i=1}^sk_iy_i = 0,\\[2mm]
    \cdots,  \\[2mm]
    \sum_{i=1}^sk_iy_i^n = 0.
  \end{array}
\right.
\end{align}

By \eqref{linear equation of eq1} and \eqref{linear equation of eq2}, one immediately has
\begin{lemma}\cite{LiWang}\label{dim_ledge}
$$\dim \overline{\mathbf{S}}(m,m-1,E^h)=(r-m-1)_+,\qquad
\dim \overline{\mathbf{S}}(n,n-1,E^v)=(s-n-1)_+.$$
Here $u_+=\max(0,u)$.
\end{lemma}



\subsection{Conformality vector spaces}

To study the dimension of spline space $\overline{\mathbf{S}}(m,n,m-1,n-1,\mathscr{T}^\varepsilon)$, we need to consider the
conformality conditions for all the (horizontal and vertical) l-edges. Thus we introduce
the following definition.

\begin{definition}\label{conformality-vector-space}
Let $E^h_i$, $i=1,2,\ldots,p$ be all the horizontal l-edges of $\mathscr{T}^\varepsilon$,
and $E^v_j$, $j=1,2,\ldots,q$ be all the vertical l-edges of $\mathscr{T}^\varepsilon$. Define
a linear space $W[\mathscr{T}^\varepsilon]$ by
$$W[\mathscr{T}^\varepsilon]:=\{\mathbf{k}=(k_1,k_2,\cdots,k_v)^T: L^h_i=0,L^v_j=0,\,i=1,\cdots,p,\,j=1,\cdots,q\},$$
where $v$ is the number vertices of $\mathscr{T}^\varepsilon$, $k_i$ is the conformality
factor corresponding to $i$-th vertex of $\mathscr{T}^\varepsilon$, and $L^h_i=0$
and $L^v_j=0$ are linear systems associated with the l-edges $E^h_i$ and $E^v_j$
respectively. $W[\mathscr{T}^\varepsilon]$ is called the conformality vector space of
$\overline{\mathbf{S}}(m,n,m-1,n-1,\mathscr{T}^\varepsilon)$. Similarly, one can define
the conformality vector space $W[E]$ of $\overline{\mathbf{S}}(m,m-1,E^h)$ (or
$\overline{\mathbf{S}}(n,n-1,E^v)$) associated with a l-edge.
\end{definition}

The following facts should be noted regarding $W[\mathscr{T}^\varepsilon]$.
\begin{enumerate}
  \item  $W[\mathscr{T}]$ can also be defined over a general T-mesh
  $\mathscr{T}$ (besides an extended T-mesh). For any $f(x,y)\in \overline{\mathbf{S}}(m,n,m-1,n-1,\mathscr{T})$,
   there is a unique corresponding vector $\mathbf{k}\in W[\mathscr{T}]$ called \textbf{conformality vector}.
  \item $k_i$, which is the $i-$th component of $\mathbf{k}$, is the conformality
  factor corresponding to $i-$th vertex in $\mathscr{T}^\varepsilon$.
  This vertex is the intersection point of two l-edges.
  So $k_i$ has to satisfy  both equations \eqref{linear equation of eq1} and
  \eqref{linear equation of eq2} associated with the two l-edges. Once $\mathbf{k}$
   is determined, the smoothing co-factors $a(y)$ and $b(x)$ can be constructed accordingly. Figure~\ref{fig:k-value interploration} illustrates an example, where  $a(y)=k_5(y-y_1)^n+k_4(y-y_2)^n+k_3(y-y_3)^n$, $b(x)=k_1(x-x_1)^m+k_2(x-x_2)^m$.
   $a(y)$ is determined by $k_3,k_4,k_5$ whose corresponding vertices lie beneath the edge associating with $a(y)$;  $b(x)$ is determined by $k_1,k_2$ whose corresponding vertices
   lie on the left of the edge corresponding to $b(x)$.

  \begin{figure}[htpb]
 \begin{center}
\begin{pspicture}(0,-1)(10,4)
\psline[linewidth=0.8pt](1,0)(5,0)
\psline[linewidth=0.8pt](1,1)(5,1)
\psline[linewidth=0.8pt](1,2)(4,2)
\psline[linewidth=0.8pt](1,3)(5,3)
\psline[linewidth=0.8pt](1,4)(5,4)
\psline[linewidth=0.8pt](1,0)(1,4)
\psline[linewidth=0.8pt](2,0)(2,4)
\psline[linewidth=0.8pt](3,0)(3,3)
\psline[linewidth=0.8pt](4,0)(4,4)
\psline[linewidth=0.8pt](5,0)(5,4)
\psline[linewidth=0.8pt](6,0)(10,0)
\psline[linewidth=0.8pt](6,1)(10,1)
\psline[linewidth=0.8pt](6,2)(9,2)
\psline[linewidth=0.8pt](6,3)(10,3)
\psline[linewidth=0.8pt](6,4)(10,4)
\psline[linewidth=0.8pt](6,0)(6,4)
\psline[linewidth=0.8pt](7,0)(7,4)
\psline[linewidth=0.8pt](8,0)(8,3)
\psline[linewidth=0.8pt](9,0)(9,4)
\psline[linewidth=0.8pt](10,0)(10,4)
 \rput(1,-0.5){$x_1$}
\rput(6,-0.5){$x_1$}
 \rput(2,-0.5){$x_2$}  \rput(7,-0.5){$x_2$}
 \rput(3,-0.5){$x_3$}\rput(8,-0.5){$x_3$}
\rput(4,-0.5){$x_4$}\rput(9,-0.5){$x_4$}
 \rput(5,-0.5){$x_5$}\rput(10,-0.5){$x_5$}
\rput(0.5,0){$y_1$}\rput(5.5,0){$y_1$}
\rput(0.5,1){$y_2$}\rput(5.5,1){$y_2$}
\rput(0.5,2){$y_3$}\rput(5.5,2){$y_3$}
\rput(0.5,3){$y_4$}\rput(5.5,3){$y_4$}
 \rput(0.5,4){$y_5$} \rput(5.5,4){$y_5$}
\rput(6,2){\color[rgb]{0.5,0.5,0.5} $k_1$}
\rput(7,2){\color[rgb]{0.5,0.5,0.5} $k_2$}
\rput(8,2){\color[rgb]{0.5,0.5,0.5} $k_3$}
\rput(9,2){\color[rgb]{0.5,0.5,0.5} $k_7$}
\rput(8,0){\color[rgb]{0.5,0.5,0.5} $k_5$}
\rput(8,1){\color[rgb]{0.5,0.5,0.5} $k_4$}
\rput(8,3){\color[rgb]{0.5,0.5,0.5} $k_6$} \rput(8,2.5){\color{red}
$a(y)$} \rput(7.5,2){\color{red} $b(x)$}
\end{pspicture}
\caption{Conformality vector}\label{fig:k-value interploration}
 \end{center}
\end{figure}
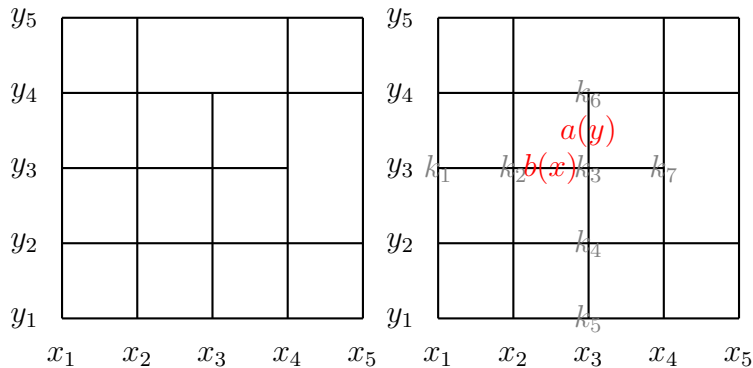


\begin{figure}[hbt]
 \begin{center}
\begin{pspicture}(0,-1)(10,5)
\psline[linewidth=0.8pt](1,0)(5,0)
\psline[linewidth=0.8pt](1,1)(5,1)
\psline[linewidth=0.8pt](1,2)(3,2)
\psline[linewidth=0.8pt,linecolor=red](3,2)(4,2)
\psline[linewidth=0.8pt](1,3)(5,3)
\psline[linewidth=0.8pt](1,4)(5,4)
\psline[linewidth=0.8pt](1,0)(1,4)
\psline[linewidth=0.8pt](2,0)(2,4)
\psline[linewidth=0.8pt](3,0)(3,3)
\psline[linewidth=0.8pt](4,0)(4,4)
\psline[linewidth=0.8pt](5,0)(5,4)
\psline[linewidth=0.8pt](6,0)(10,0)
\psline[linewidth=0.8pt](6,1)(10,1)
\psline[linewidth=0.8pt](6,2)(8,2)
\psline[linewidth=0.8pt,linecolor=red,linestyle=dashed](8,2)(9,2)
\psline[linewidth=0.8pt](6,3)(10,3)
\psline[linewidth=0.8pt](6,4)(10,4)
\psline[linewidth=0.8pt](6,0)(6,4)
\psline[linewidth=0.8pt](7,0)(7,4)
\psline[linewidth=0.8pt](8,0)(8,3)
\psline[linewidth=0.8pt](9,0)(9,4)
\psline[linewidth=0.8pt](10,0)(10,4)
 \rput(1,-0.5){$x_1$}
\rput(6,-0.5){$x_1$}
 \rput(2,-0.5){$x_2$}  \rput(7,-0.5){$x_2$}
 \rput(3,-0.5){$x_3$}\rput(8,-0.5){$x_3$}
\rput(4,-0.5){$x_4$}\rput(9,-0.5){$x_4$}
 \rput(5,-0.5){$x_5$}\rput(10,-0.5){$x_5$}
\rput(0.5,0){$y_1$}\rput(5.5,0){$y_1$}
\rput(0.5,1){$y_2$}\rput(5.5,1){$y_2$}
\rput(0.5,2){$y_3$}\rput(5.5,2){$y_3$}
\rput(0.5,3){$y_4$}\rput(5.5,3){$y_4$}
 \rput(0.5,4){$y_5$} \rput(5.5,4){$y_5$}
 \rput(3,4.5){$\mathscr{T}$}\rput(8,4.5){$\mathscr{T}'$}
\rput(6,2){\color[rgb]{0.5,0.5,0.5} $k_1$}
\rput(1,2){\color[rgb]{0.5,0.5,0.5} $k_1$}
\rput(7,2){\color[rgb]{0.5,0.5,0.5} $k_2$}
\rput(2,2){\color[rgb]{0.5,0.5,0.5} $k_2$}
\rput(8,2){\color[rgb]{0.5,0.5,0.5} $k_3$}
\rput(3,2){\color[rgb]{0.5,0.5,0.5} $k_3$}
\rput(4,2){\color[rgb]{0.5,0.5,0.5} $k_7$}
\rput(8,0){\color[rgb]{0.5,0.5,0.5} $k_5$}
\rput(3,0){\color[rgb]{0.5,0.5,0.5} $k_5$}
\rput(8,1){\color[rgb]{0.5,0.5,0.5} $k_4$}
\rput(3,1){\color[rgb]{0.5,0.5,0.5} $k_4$}
\rput(8,3){\color[rgb]{0.5,0.5,0.5} $k_6$}
\rput(3,3){\color[rgb]{0.5,0.5,0.5} $k_6$}
\end{pspicture}
\caption{The naught value of a conformality factor.}\label{fig:T and T'}
 \end{center}
\end{figure}
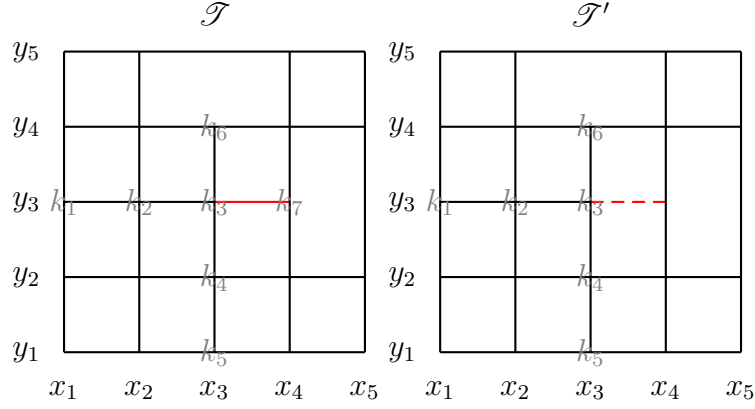

\item  A naught value of $k_i$ which is the conformality factor corresponding to a T-vertex
 results in the vanishing of the corresponding vertex in the T-mesh $\mathscr{T}$.
 An example is illustrated in Figure~\ref{fig:T and  T'} where $k_7=0$,
 and in this case $\overline{\mathbf{S}}(m,n,m-1,n-1,\mathscr{T})=\overline{\mathbf{S}}(m,n,m-1,n-1,\mathscr{T}')$,
 or $\mathscr{T}$ degenerates to $\mathscr{T}'$.
\end{enumerate}

 As an example, we discuss the conformality
vectors of B-spline functions.

\begin{example}\label{example:Bspline curve }
Let $\Delta:-\infty<x_1<x_2<\ldots<x_r<\infty$ ($r>m+1$) be a
partition of $\mathbb{R}$, and $\overline{\mathbf{S}}(m,m-1,\Delta)$
be the B-spline function space defined over $\Delta$, that is,
$\overline{\mathbf{S}}(m,m-1,\Delta)$ consists of piecewise
polynomials of degree $m$ with $C^{m-1}$ continuity over $\Delta$.
The B-spine basis function $N[x_i,x_{i+1},\ldots,x_{i+m+1}]$ is an
element in $\overline{\mathbf{S}}(m,m-1,E_i)$, where $E_i$ is the
interval $[x_i,x_{i+m+1}]$ with the partition:
$x_i<x_{i+1}<\ldots<x_{i+m+1}$. Here $1\le i\le r-m-1$. The
conformality vector $\mathbf{k}=(k_i,k_{i+1},\ldots,k_{i+m+1})$ of
$N[x_i,x_{i+1},\ldots,x_{i+m+1}]$ can be obtained by solving the
associated linear system \eqref{linear equation of eq1}. It is easy
to see that $k_jk_{j+1}<0$, $j=i,i+1,\ldots,i+m$ and $k_i>0$.
\end{example}


\begin{example}
\label{example:Bspline surface} Let
$\overline{\mathbf{S}}(m,m-1,\Delta_x)$ be the B-spline function
space defined over $\Delta_x:-\infty<x_1<x_2<\ldots<x_r<\infty$
($r>m+1$), and $\overline{\mathbf{S}}(n,n-1,\Delta_y)$ be the
B-spline function space defined over
$\Delta_y:-\infty<y_1<y_2<\ldots<y_s<\infty$ ($s>n+1$). We denote
$\mathscr{T}_\otimes$ as a tensor product mesh $\Delta_x\times
\Delta_y$.
It is easy to see that, $f(x)\in
\overline{\mathbf{S}}(m,m-1,\Delta_x)$ and $g(y)\in
\overline{\mathbf{S}}(n,n-1,\Delta_y)$ implies that $f(x)g(y) \in
\overline{\mathbf{S}}(m,n,m-1,n-1,\mathscr{T}_\otimes)$.

Now assume that the B-spline basis functions
$N[x_i,x_{i+1},\ldots,x_{i+m+1}]\in
\overline{\mathbf{S}}(m,m-1,\Delta_x)$ and
$N[y_j,y_{j+1},\ldots,y_{j+n+1}]\in
\overline{\mathbf{S}}(n,n-1,\Delta_y)$ have conformality vectors
${\bf k}^1:=(k_i^1,k_{i+1}^1,\ldots,k_{i+m+1}^1)$ and ${\bf
k}^2:=(k_j^2,k_{j+1}^2,\ldots,k_{j+n+1}^2)$ respectively. Then by
equation~\eqref{kvalue}, the B-spline basis function
$N[x_i,x_{i+1},\ldots,x_{i+m+1}]\cdot
N[y_j,y_{j+1},\ldots,y_{j+n+1}]\in
\overline{\mathbf{S}}(m,n,m-1,n-1,\mathscr{T}_\otimes)$ has a
conformality vector ${\bf k}^1 \otimes {\bf k}^2$ which is a vector
of dimension $(m+2)(n+2)$ with elements $k_p^1k_q^2$,
$p=i,i+1,\ldots,i+m+1,\,q=j,j+1,\ldots,j+n+1$.
\end{example}




\subsection{Equivalence of spline spaces}

Based on the above preparations, we obtain a mapping $\mathscr{K}$  between
the spline space $\overline{\mathbf{S}}(m,n,m-1,n-1,\mathscr{T}^\varepsilon)$ and the
conformality vector space $W[\mathscr{T}^\varepsilon]$:
\begin{equation}\mathscr{K}: \overline{\mathbf{S}}(m,n,m-1,n-1,\mathscr{T}^\varepsilon)\longrightarrow W[\mathscr{T}^\varepsilon].\end{equation}

By Equation \eqref{kvalue}, $\mathscr{K}$ is a linear mapping due to the linear
property of the operator $\partial^{m+n}~/{\partial x^m\partial y^n}$. In fact,
$\mathscr{K}$ is an isomorphic mapping.
\begin{theorem}
\label{isohomophism} The mapping
  $\mathscr{K}: \overline{\mathbf{S}}(m,n,m-1,n-1,\mathscr{T}^\varepsilon)\longrightarrow W[\mathscr{T}^\varepsilon]$
is bijective.
\end{theorem}

\begin{proof} We first prove that $\mathscr{K}$ is injective. It is enough to show that
$\mathscr{K}f={\bf 0}$ implies $f\equiv 0$ for $f\in \overline{\mathbf{S}}(m,n,m-1,n-1,\mathscr{T}^\varepsilon)$. By the remark 3 following {\bf Definition}~\ref{conformality-vector-space}, if the conformality vector of $f(x,y)$ is a
zero vector, then $f(x,y)$ is a single polynomial over $\mathscr{T}^\varepsilon$. By the homogeneous  boundary conditions, $f(x,y)\equiv 0$. Thus $\mathscr{K}$ is injective.

Next we show $\mathscr{K}$ is surjective. By the remark 2 following {\bf Definition}~\ref{conformality-vector-space}, for a given conformality vector ${\bf k}\in
W[\mathscr{T}^\varepsilon]$, one can construct a smoothing cofactor for each edge of the T-mesh $\mathscr{T}^\varepsilon$, and thus obtains a spline function $f(x,y)\in \overline{\mathbf{S}}(m,n,m-1,n-1,\mathscr{T}^\varepsilon)$ corresponding to ${\bf k}$, that is, $\mathscr{K}$ is surjective. Thus the mapping $\mathscr{K}$ is bijective.
\end{proof}

By the above theorem, the spline space
$\overline{\mathbf{S}}(m,n,m-1,n-1,\mathscr{T}^\varepsilon)$ is
isomorphic to the conformality vector space $W[\mathscr{T}^\varepsilon]$, and
\begin{equation}\label{isomorphic_tspline}
\dim\overline{\mathbf{S}}(m,n,m-1,n-1,\mathscr{T}^\varepsilon)=\dim W[\mathscr{T}^\varepsilon].
\end{equation}

Similarly, one can show that
\begin{equation}\label{isomorphic_ledge}
\overline{\mathbf{S}}(m,m-1,E^h) \cong W[E^h], \qquad
  \overline{\mathbf{S}}(n,n-1,E^v) \cong W[E^v]
\end{equation}
for a horizontal l-edge $E^h$ and a vertical l-edge $E^v$.

In the following, we only have to analyze the structure of $W[\mathscr{T}^\varepsilon].$

\section{The dimension formula}
In this section, we will derive a dimension formula for the spline
$\overline{\mathbf{S}}(m,n,m-1,n-1,\mathscr{T}^\varepsilon)$ over a certain type
of hierarchical T-meshes by direct sum decomposition of the spline space.
The decomposition is according to each interior l-edge of the T-mesh $\mathscr{T}^\varepsilon$. A set of basis functions of the spline space is also constructed.
\subsection{Some definitions}
We first introduce a certain type of hierarchical T-meshes associated with
degree $m, n$ and denote it by $\mathscr{T}_{m,n}$.

\begin{definition}
  Let
  $\mathscr{T}_\otimes=[x_0,x_1,\cdots,x_p]\times[y_0,y_1,\cdots,y_q]$
  be a tensor product mesh, and let
  $X=\{x_0,x_{m-1},\cdots,x_{s(m-1)},x_p\}$ and $Y=\{y_0,y_{n-1},\cdots,y_{t(n-1)},y_q\}$,
  where $s$ and $t$ are unique integers such that $0<p-s(m-1)\leq m-1,\,0<q-t(n-1)\leq
  n-1$. The domain $\Omega=[x_0,x_p]\times[y_0,y_q]$ is subdivided by the lines
  $\{x=x_i,y=y_j:x_i\in X,y_j\in Y\}$ into sub-domains, each of which is occupied by a local
   tensor product mesh of size $(m-1)\times (n-1)$ (or smaller near the right or upper boundary lines).
    Each subdomain is called a $(m,n)$-subdomain of $\mathscr{T}_\otimes$.
a $(m,n)$-subdomain is called s boundary one, if it is near the
boundary lines of T-mesh.
    A subdomain is said subdivided if each cell in
  the subdomain is subdivided (a cell is subdivided by connecting the middle points of oppose sides
  of the cell with two line segments). A subdomain is called isolated if it is subdivided while its adjacent subdomains
  (two subdomains are called adjacent if they share a common boundary line segment) are not subdivided and it is not boundary one.
\end{definition}

Figure~\ref{fig:process of subdivision}(a) illustrates an example
where a tensor product mesh is subdivided into $(3,3)$-subdomains by
blue lines. In Figure~\ref{fig:process of subdivision}(b), two
subdomains are further subdivided by red lines, and these one
subdomain is isolated.


 \begin{figure}[hbpt]
 \begin{center}
 \begin{pspicture}(0,0)(18,7)
 \psset{unit=0.8cm}
   \psline[linewidth=0.8pt](1,1)(6,1)
   \psline[linewidth=0.8pt](1,2)(6,2)
   \psline[linewidth=0.8pt,linecolor=blue](1,3)(6,3)
   \psline[linewidth=0.8pt](1,4)(6,4)
   \psline[linewidth=0.8pt,linecolor=blue](1,5)(6,5)
   \psline[linewidth=0.8pt](1,6)(6,6)
   \psline[linewidth=0.8pt](1,7)(6,7)
   \psline[linewidth=0.8pt](1,1)(1,7)
   \psline[linewidth=0.8pt](2,1)(2,7)
   \psline[linewidth=0.8pt,linecolor=blue](3,1)(3,7)
   \psline[linewidth=0.8pt](4,1)(4,7)
   \psline[linewidth=0.8pt,linecolor=blue](5,1)(5,7)
   \psline[linewidth=0.8pt](6,1)(6,7)

       \rput(3.5,0){(a)}

   \psline[linewidth=0.8pt](7,1)(12,1)
   \psline[linewidth=0.8pt](7,2)(12,2)
   \psline[linewidth=0.8pt,linecolor=blue](7,3)(12,3)
   \psline[linewidth=0.8pt](7,4)(12,4)
   \psline[linewidth=0.8pt,linecolor=blue](7,5)(12,5)
   \psline[linewidth=0.8pt](7,6)(12,6)
   \psline[linewidth=0.8pt](7,7)(12,7)
   \psline[linewidth=0.8pt](7,1)(7,7)
   \psline[linewidth=0.8pt](8,1)(8,7)
   \psline[linewidth=0.8pt,linecolor=blue](9,1)(9,7)
   \psline[linewidth=0.8pt](10,1)(10,7)
   \psline[linewidth=0.8pt,linecolor=blue](11,1)(11,7)
   \psline[linewidth=0.8pt](12,1)(12,7)

   \rput(9.5,0){(b)}
   \psline[linewidth=0.8pt,linecolor=red](9,3.5)(11,3.5)
   \psline[linewidth=0.8pt,linecolor=red](9,4.5)(11,4.5)
   \psline[linewidth=0.8pt,linecolor=red](11,2.5)(12,2.5)
   \psline[linewidth=0.8pt,linecolor=red](11,1.5)(12,1.5)
   \psline[linewidth=0.8pt,linecolor=red](9.5,3)(9.5,5)
   \psline[linewidth=0.8pt,linecolor=red](10.5,3)(10.5,5)
   \psline[linewidth=0.8pt,linecolor=red](11.5,1)(11.5,3)
   \psline[linewidth=0.8pt](13,1)(18,1)
   \psline[linewidth=0.8pt](13,2)(18,2)
   \psline[linewidth=0.8pt,linecolor=blue](13,3)(18,3)
   \psline[linewidth=0.8pt](13,4)(18,4)
   \psline[linewidth=0.8pt,linecolor=blue](13,5)(18,5)
   \psline[linewidth=0.8pt](13,6)(18,6)
   \psline[linewidth=0.8pt](13,7)(18,7)
   \psline[linewidth=0.8pt](13,1)(13,7)
   \psline[linewidth=0.8pt](14,1)(14,7)
   \psline[linewidth=0.8pt,linecolor=blue](15,1)(15,7)
   \psline[linewidth=0.8pt](16,1)(16,7)
   \psline[linewidth=0.8pt,linecolor=blue](17,1)(17,7)
   \psline[linewidth=0.8pt](18,1)(18,7)
   \psline[linewidth=0.8pt](15,3.5)(17,3.5)
   \psline[linewidth=0.8pt](15,4.5)(17,4.5)
   \psline[linewidth=0.8pt](17,2.5)(18,2.5)
   \psline[linewidth=0.8pt](17,1.5)(18,1.5)
   \psline[linewidth=0.8pt](15.5,3)(15.5,5)
   \psline[linewidth=0.8pt](16.5,3)(16.5,5)
   \psline[linewidth=0.8pt](17.5,1)(17.5,3)
   \psline[linewidth=0.8pt,linecolor=red](15.25,4)(15.25,5)
   \psline[linewidth=0.8pt,linecolor=red](15.75,4)(15.75,5)
   \psline[linewidth=0.8pt,linecolor=red](16.25,3)(16.25,4)
   \psline[linewidth=0.8pt,linecolor=red](16.75,3)(16.75,4)
   \psline[linewidth=0.8pt,linecolor=red](17.25,1)(17.25,2)
   \psline[linewidth=0.8pt,linecolor=red](17.75,1)(17.75,2)
   \psline[linewidth=0.8pt,linecolor=red](17,1.25)(18,1.25)
   \psline[linewidth=0.8pt,linecolor=red](17,1.75)(18,1.75)
   \psline[linewidth=0.8pt,linecolor=red](16,3.25)(17,3.25)
   \psline[linewidth=0.8pt,linecolor=red](16,3.75)(17,3.75)
   \psline[linewidth=0.8pt,linecolor=red](15,4.25)(16,4.25)
   \psline[linewidth=0.8pt,linecolor=red](15,4.75)(16,4.75)
   \rput(15.5,0){(c)}
 \end{pspicture}
  \caption{Hierarchical T-mesh $\mathscr{T}_{m,n}$ in the case of $m=n=3,p=5,q=6$}.\label{fig:process of subdivision}
 \end{center}
\end{figure}
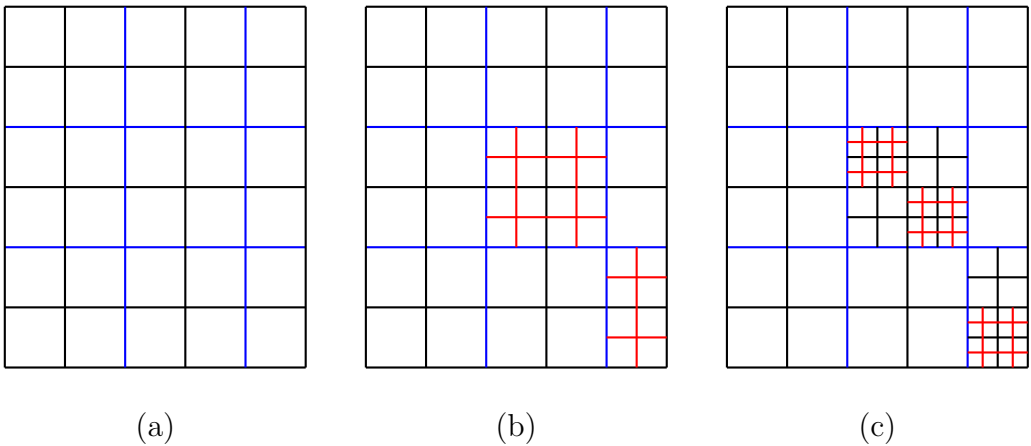

\begin{definition}
A hierarchical T-mesh $\mathscr{T}_{m,n}$ associated with $(m,n)$ is defined level by level. At level $k=0$, $\mathscr{T}_{m,n}$ is a tensor product mesh $\mathscr{T}_\otimes$. From
level $0$ to level $1$, subdivide $\mathscr{T}_\otimes$ into $(m,n)$-subdomains, and
some of which are further subdivided to get the mesh at level $1$. Generally from level $k$ to level $k+1$, subdivide some local tensor product meshes at level $k$ (which are obtained by subdividing subdomains at level $k-1$) into $(m,n)$-subdomains at level $k$, some
of which are further subdivided to get level $k+1$ mesh.
\end{definition}

The T-mesh $\mathscr{T}_{m,n}$ is a special type of hierarchical T-mesh~\cite{Deng2006PHT}.
In particular, it is a regular hierarchical T-mesh when $m=n=2$. Figure~\ref{fig:process of subdivision} illustrates the process of generating a hierarchical T-mesh $\mathscr{T}_{3,3}$.
Figure~\ref{fig:process of subdivision}(a) shows that a tensor product mesh is subdivided into
subdomains at level $k=0$. Figure~\ref{fig:process of subdivision}(b) shows two subdomains at level $k=0$ are subdivided to get a mesh at level $k=1$. Figure~\ref{fig:process of subdivision}(c) shows two local tensor product meshes (corresponding to the two subdivided subdomains in Figure~\ref{fig:process of subdivision}(b)) are subdivided into $(3,3)$-subdomains at level $k=1$, and three subdomains at level $k=1$ are further subdivided to get a mesh at level $k=2$.

\subsection{The dimension formulas}

Now we are able to state our main result in the paper -- the
dimension formula for spline space
$\overline{\mathbf{S}}(m,n,m-1,n-1,\mathscr{T}_{m,n}^\varepsilon)$,
where $\mathscr{T}_{m,n}^\varepsilon$ is $\mathscr{T}_{m,n}$'s
extended T-mesh associated with degree $m,n$.

\begin{theorem}\label{theorem_dimension_formula}
Let $V^+$ be the number of crossing vertices of $\mathscr{T}_{m,n}^\varepsilon$.
Assume $\mathscr{T}_{m,n}^\varepsilon$ has $E_H$ interior horizontal l-edges
and $E_V$ interior vertical l-edges. Then
\begin{align}\label{dimension_formula}
   & \dim\overline{\mathbf{S}}(m,n,m-1,n-1,\mathscr{T}_{m,n}^\varepsilon)\nonumber\\
   & =(m-1)(n-1)+V^+-(m-1)E_H-(n-1)E_V+\delta,
\end{align}
where $\delta$ is the number of isolated subdomains of
$\mathscr{T}_{m,n}^\varepsilon$ at all levels.
\end{theorem}

Before proving the theorem, we need some preparations.

\begin{lemma}
\label{lemma:surjection}
  Let $\mathscr{T}_{m,n}$ be a hierarchical T-mesh associated with $(m,n)$, and
  $\mathscr{T}_{m,n}^\varepsilon$ be its extended T-mesh. $\mathbf{E}$ is the set
  of all the interior l-edges of $\mathscr{T}_{m,n}^\varepsilon$. Then the l-edges
  in $\mathbf{E}$ can be ordered as $E_t\prec E_{t-1}\prec \cdots \prec E_1$ such that
  the projection mapping
  $$\pi:\,\, W[\mathscr{T}_i]\longrightarrow W[E_i]$$
  is surjective, $i=1,2,\ldots,t$. Here $\mathscr{T}_i$ is the T-mesh obtained by deleting
  $\{E_1,\cdots,E_{i-1}\}$ from $\mathscr{T}_{m,n}^\varepsilon$, and $E_i\subset
  \mathscr{T}_i$ is regarded as a l-edge in $\mathscr{T}_i$.
  \end{lemma}
\begin{proof}
See the appendix.
\end{proof}

Base on the above lemma, we can decompose the conformality vector space
$W[\mathscr{T}^\varepsilon_{m,n}]$ into direct sums of $W[E_i]$.

\begin{theorem}
  \label{main therom}
Let the notations be the same as in Lemma~\ref{lemma:surjection}. Then
$$W[\mathscr{T}_{i}]\cong W[\mathscr{T}_{i+1}]\oplus W[E_i], \quad i=1,2,\ldots,t.$$
\end{theorem}
\begin{proof} Since $\mathscr{T}_{i}= \mathscr{T}_{i+1}\cup E_i$,
$\overline{\mathbf{S}}(m,n,m-1,n-1,\mathscr{T}_{i+1}) \subset
\overline{\mathbf{S}}(m,n,m-1,n-1,\mathscr{T}_{i})$. Correspondingly,
$W[\mathscr{T}_{i+1}]$ can be regarded as a subspace of $W[\mathscr{T}_i]$ by
taking every vector $\mathbf k\in W[\mathscr{T}_{i+1}]$ as a vector in
$W[\mathscr{T}_i]$ whose components corresponding to the vertices of $E_i$ are zero,
and the remaining components are the same as ${\mathbf k}'s$.

Now consider the sequence
\begin{equation}\label{exact sequence}
\xymatrix@C=0.5cm{
  0 \ar[r] & W[\mathscr{T}_{i+1}] \ar[rr]^{i} && W[\mathscr{T}_i] \ar[rr]^{\pi} && W[E_i] \ar[r] & 0
  },
\end{equation}
where $i$ in the natural embedded mapping, and $\pi$ is the projection mapping defined
in Lemma~\ref{lemma:surjection}. We will show that the sequence is exact. Since $i$ is injective and $\pi$ is surjective by Lemma~\ref{lemma:surjection}, it is enough to show
that $Im(i)=Ker(\pi)$.


On one hand, for any ${\mathbf k}\in W[\mathscr{T}_{i+1}]$, $i({\mathbf k})={\bf k'}$ is a
vector in $W[\mathscr{T}_{i}]$ whose components (conformality factors) corresponding to
the vertices of $E_i$ are zero. Thus $\pi(i({\mathbf k}))={\bf 0} \in W[E_i]$, that is,
$Im(i)\subset Ker(\pi)$.

On the other hand, for any $\mathbf{k}\in Ker(\pi)$, $\pi(\mathbf{k})={\bf 0}$, that is,
the components of $\mathbf{k}$ corresponding to the vertices of $E_i$ are zero. Therefore
$\mathbf{k}\in Im(i)$, i.e., $Ker(\pi)\subset Im(i)$.


Thus the sequence \eqref{exact sequence} is exact. Since $W[\mathscr{T}_{i+1}]$, $W[\mathscr{T}_{i}]$ and $W[E_i]$ are linear spaces, the sequence is split, and therefore
$W[\mathscr{T}_{i}]\cong W[\mathscr{T}_{i+1}]\oplus W[E_i]$, as required.
\end{proof}


Based on the above theorem, we have
\begin{corollary}
The set of functions
$\cup_{i=1}^{t}\{B_j^i(x,y)\}_{j=1}^{j=r_i-m-1}$ defined in the
proof of Lemma~ \ref{lemma:surjection} form a basis for the spline
space
$\overline{\mathbf{S}}(m,n,m-1,n-1,\mathscr{T}_{m,n}^\varepsilon)$.
\end{corollary}

Now we are ready to prove the dimension formula \eqref{dimension_formula}.

\begin{proof}[Proof of Theorem~\ref{theorem_dimension_formula}]

   By Theorem~\ref{main therom}, $W[\mathscr{T}_{m,n}^\varepsilon]$ can be decomposed
   into direct sums of $W[E_i]$:
   $$W[\mathscr{T}_{m,n}^\varepsilon]\cong\oplus_{i=1}^t W[E_i].$$

    The interior l-edges $\mathbf{E}=\{E_1,E_2,\ldots,E_t\}$ of $\mathscr{T}_{m,n}^\varepsilon$ can be divided into disjoint sets ${\mathbf E}^i$, $i=0,1,\ldots,l$, i.e., ${\mathbf E}=\cup_{i=0}^l {\mathbf E}^i$ and ${\mathbf E}^i\cap {\mathbf E}^j=\emptyset,\,i\not=j$. ${\mathbf E}^i$ is the set of interior l-edges of $\mathscr{T}_{m,n}^\varepsilon$ defined at level $i$, $i=0,1,\ldots,l$.
   In particular, ${\mathbf E}^0$ is the set of the interior l-edges of the initial extended tensor product mesh  $\mathscr{T}_\otimes^\varepsilon$.
   By the order defined in Lemma~\ref{lemma:surjection}, any element in ${\mathbf E}^i$
   precedes any element in ${\mathbf E}^j$ for $i>j$.

   By the direct sum decomposition of $W[\mathscr{T}_{m,n}^\varepsilon]$,
   \begin{equation}\label{dimension_sum}
   \dim W[\mathscr{T}_{m,n}^\varepsilon]=\sum_{i=0}^l\sum_{E_j\in {\mathbf E}^i}\dim W[E_j].
   \end{equation}

   Now we count $\dim W[E_j]$ according to the level of $E_j$ from the highest level $l$
   to the lowest level $0$.

   By Lemma~\ref{dim_ledge} and \eqref{isomorphic_ledge},
   $$ \dim W[E_j]=\dim\overline{\mathbf{S}}(m,m-1,E_j)=(v_j^+-m+1)+$$
   if $E_j$ is a horizontal l-edge, and
   $$ \dim W[E_j]=\dim\overline{\mathbf{S}}(n,n-1,E_j)=(v_j^+-n+1)+$$
   if $E_j$ is a vertical l-edge. Here $v_j^+$ is the number of crossing vertices of
   $E_j$ in the mesh $\mathscr{T}_j$.

   Consider all the interior l-edges ${\mathbf E}^i$ in a fixed level $i$ ($i=l,l-1,\ldots,0$). ${\mathbf E}^i$ is obtained by subdividing some subdomains at level $i-1$ (except for $i=0$ where ${\mathbf E}^0$ is the initial extended tensor product mesh). As shown in the proof of Lemma~\ref{lemma:surjection}, if there is no isolated subdomain at level $i-1$,
   then one can order the l-edges in ${\mathbf E}^i$ such that
   $$
   v^+_j\ge \left\{\begin{array}{ll} m-1, & \mbox{\quad if $E_j$ is a horizontal l-edge}\\
                                n-1, & \mbox{\quad if $E_j$ is a vertical l-edge}
   \end{array}
      \right.
   $$
     If a subdomain $\mathscr{S}$ is isolated, then one can order the l-edges in the subdomain $\mathscr{S}$ such that the above equality holds for all but one vertical l-edge in $\mathscr{S}$.
   For the exceptional l-edge,
   $$
   v^+_j=n-2.
   $$

  Figure \ref{fig:edge-seq}, illustrates an example for the order of
  l-edges in $\mathbf{E}^2$ in the case of $m=4,n=3$.

   From the above facts, one has
   \begin{equation}\label{dimension_level_i}
   \sum_{E_j\in {\mathbf E}^i}\dim W[E_j]=\sum_{E_j\in {\mathbf E}^i_h} (v_j^+-m+1)+\sum_{E_j\in {\mathbf E}^i_v} (v_j^+-n+1)+\delta_i, \quad i\ge 1,
   \end{equation}
   where ${\mathbf E}^i_h$ is the set of horizontal l-edges in ${\mathbf E}^i$, ${\mathbf E}^i_v$ is the set of vertical l-edges in ${\mathbf E}^i$, and $\delta_i$ is the number of isolated subdomains at level $i-1$. For $i=0$,
   since ${\mathbf E}^0$ is the set of l-edges of a tensor product mesh $\mathscr{T}_\otimes^\varepsilon$, it is direct to check that
       \begin{equation}\label{dimension_level_0}
   \sum_{E_j\in {\mathbf E}^0}\dim W[E_j]=\sum_{E_j\in {\mathbf E}^0_h} (v_j^+-m+1)+\sum_{E_j\in {\mathbf E}^0_v} (v_j^+-n+1)+(m-1)(n-1).
   \end{equation}

   Now the dimension formula \eqref{dimension_formula} follows from \eqref{isomorphic_tspline}, \eqref{dimension_sum}, \eqref{dimension_level_i} and \eqref{dimension_level_0}.
\end{proof}

\subsection{Examples}
In this subsection, we give two examples to count the dimensions of the spline spaces
$\overline{\mathbf{S}}(m,n,m-1,n-1,\mathscr{T}_{m,n}^\varepsilon)$.

\begin{example}
  In the case of $m=2,n=2$, $\mathscr{T}_{m,n}$ is a regular
  hierarchical T-mesh and we denote it as $\mathscr{T}$. By the dimension formula \eqref{dimension_formula},
  \begin{align*}
  \dim\mathbf{S}(2,2,1,1,\mathscr{T})&=\dim\overline{\mathbf{S}}(2,2,1,1,\mathscr{T}^\varepsilon)\\
  &=V^+-E+\delta+1
  \end{align*}
  where, $V^+$ is the number of crossing vertices,
  $E$ is the number of interior l-edges, and $\delta$ is the number
  of isolated cells of $\mathscr{T}^\varepsilon$. This formula has
  been presented in \cite{Deng2008:21}.
\end{example}

 Next we give a concrete example for the case $m=3,n=3$.
 \begin{example}
 In this example, the hierarchical T-mesh $\mathscr{T}_{3,3}$ (which is the same as the
 T-mesh in Figure~\ref{fig:process of subdivision}(c)) and its extension
 $\mathscr{T}_{3,3}^\varepsilon$ are illustrated in  Figure~\ref{fig:extensionT33}.
 \begin{figure}[htpb]
 \psset{xunit=0.6cm,yunit=0.6cm}
 \centering
\begin{pspicture}(6,-1)(19,8)%
 \psline[linewidth=0.8pt](6,1)(11,1)
   \psline[linewidth=0.8pt](6,2)(11,2)
   \psline[linewidth=0.8pt,linecolor=blue](6,3)(11,3)
   \psline[linewidth=0.8pt](6,4)(11,4)
   \psline[linewidth=0.8pt,linecolor=blue](6,5)(11,5)
   \psline[linewidth=0.8pt](6,6)(11,6)
   \psline[linewidth=0.8pt](6,7)(11,7)
   \psline[linewidth=0.8pt](6,1)(6,7)
   \psline[linewidth=0.8pt](7,1)(7,7)
   \psline[linewidth=0.8pt,linecolor=blue](8,1)(8,7)
   \psline[linewidth=0.8pt](9,1)(9,7)
   \psline[linewidth=0.8pt,linecolor=blue](10,1)(10,7)
   \psline[linewidth=0.8pt](11,1)(11,7)
   \psline[linewidth=0.8pt](8,3.5)(10,3.5)
   \psline[linewidth=0.8pt](8,4.5)(10,4.5)
   \psline[linewidth=0.8pt](10,2.5)(11,2.5)
   \psline[linewidth=0.8pt](10,1.5)(11,1.5)
   \psline[linewidth=0.8pt](8.5,3)(8.5,5)
   \psline[linewidth=0.8pt](9.5,3)(9.5,5)
   \psline[linewidth=0.8pt](10.5,1)(10.5,3)
   \psline[linewidth=0.8pt,linecolor=red](8.25,4)(8.25,5)
   \psline[linewidth=0.8pt,linecolor=red](8.75,4)(8.75,5)
   \psline[linewidth=0.8pt,linecolor=red](9.25,3)(9.25,4)
   \psline[linewidth=0.8pt,linecolor=red](9.75,3)(9.75,4)
   \psline[linewidth=0.8pt,linecolor=red](10.25,1)(10.25,2)
   \psline[linewidth=0.8pt,linecolor=red](10.75,1)(10.75,2)
   \psline[linewidth=0.8pt,linecolor=red](10,1.25)(11,1.25)
   \psline[linewidth=0.8pt,linecolor=red](10,1.75)(11,1.75)
   \psline[linewidth=0.8pt,linecolor=red](9,3.25)(10,3.25)
   \psline[linewidth=0.8pt,linecolor=red](9,3.75)(10,3.75)
   \psline[linewidth=0.8pt,linecolor=red](8,4.25)(9,4.25)
   \psline[linewidth=0.8pt,linecolor=red](8,4.75)(9,4.75)
   \psline[linewidth=0.8pt](13,1)(18,1)
   \psline[linewidth=0.8pt](13,2)(18,2)
   \psline[linewidth=0.8pt,linecolor=blue](13,3)(18,3)
   \psline[linewidth=0.8pt](13,4)(18,4)
   \psline[linewidth=0.8pt,linecolor=blue](13,5)(18,5)
   \psline[linewidth=0.8pt](13,6)(18,6)
   \psline[linewidth=0.8pt](13,7)(18,7)
   \psline[linewidth=0.8pt](13,1)(13,7)
   \psline[linewidth=0.8pt](14,1)(14,7)
   \psline[linewidth=0.8pt,linecolor=blue](15,1)(15,7)
   \psline[linewidth=0.8pt](16,1)(16,7)
   \psline[linewidth=0.8pt,linecolor=blue](17,1)(17,7)
   \psline[linewidth=0.8pt](18,1)(18,7)
   \psline[linewidth=0.8pt](15,3.5)(17,3.5)
   \psline[linewidth=0.8pt](15,4.5)(17,4.5)
   \psline[linewidth=0.8pt](17,2.5)(18,2.5)
   \psline[linewidth=0.8pt](17,1.5)(18,1.5)
   \psline[linewidth=0.8pt](15.5,3)(15.5,5)
   \psline[linewidth=0.8pt](16.5,3)(16.5,5)
   \psline[linewidth=0.8pt](17.5,1)(17.5,3)
   \psline[linewidth=0.8pt,linecolor=red](15.25,4)(15.25,5)
   \psline[linewidth=0.8pt,linecolor=red](15.75,4)(15.75,5)
   \psline[linewidth=0.8pt,linecolor=red](16.25,3)(16.25,4)
   \psline[linewidth=0.8pt,linecolor=red](16.75,3)(16.75,4)
   \psline[linewidth=0.8pt,linecolor=red](17.25,1)(17.25,2)
   \psline[linewidth=0.8pt,linecolor=red](17.75,1)(17.75,2)
   \psline[linewidth=0.8pt,linecolor=red](17,1.25)(18,1.25)
   \psline[linewidth=0.8pt,linecolor=red](17,1.75)(18,1.75)
   \psline[linewidth=0.8pt,linecolor=red](16,3.25)(17,3.25)
   \psline[linewidth=0.8pt,linecolor=red](16,3.75)(17,3.75)
   \psline[linewidth=0.8pt,linecolor=red](15,4.25)(16,4.25)
   \psline[linewidth=0.8pt,linecolor=red](15,4.75)(16,4.75)
   \psline[linewidth=0.8pt,linecolor=gray](12,0)(12,8)
   \psline[linewidth=0.8pt,linecolor=gray](12.33,0)(12.33,8)
   \psline[linewidth=0.8pt,linecolor=gray](12.66,0)(12.66,8)
   \psline[linewidth=0.8pt,linecolor=gray](18.33,0)(18.33,8)
   \psline[linewidth=0.8pt,linecolor=gray](18.66,0)(18.66,8)
   \psline[linewidth=0.8pt,linecolor=gray](19,0)(19,8)
   \psline[linewidth=0.8pt,linecolor=gray](12,0)(19,0)
   \psline[linewidth=0.8pt,linecolor=gray](12,0.33)(19,0.33)
   \psline[linewidth=0.8pt,linecolor=gray](12,0.66)(19,0.66)
   \psline[linewidth=0.8pt,linecolor=gray](12,7.33)(19,7.33)
   \psline[linewidth=0.8pt,linecolor=gray](12,7.66)(19,7.66)
   \psline[linewidth=0.8pt,linecolor=gray](12,8)(19,8)
   \psline[linewidth=0.8pt,linecolor=gray](13,1)(13,0)
   \psline[linewidth=0.8pt,linecolor=gray](14,1)(14,0)
   \psline[linewidth=0.8pt,linecolor=gray](15,1)(15,0)
   \psline[linewidth=0.8pt,linecolor=gray](16,1)(16,0)
   \psline[linewidth=0.8pt,linecolor=gray](17,1)(17,0)
   \psline[linewidth=0.8pt,linecolor=gray](18,1)(18,0)
   \psline[linewidth=0.8pt,linecolor=gray](13,7)(13,8)
   \psline[linewidth=0.8pt,linecolor=gray](14,7)(14,8)
   \psline[linewidth=0.8pt,linecolor=gray](15,7)(15,8)
   \psline[linewidth=0.8pt,linecolor=gray](16,7)(16,8)
   \psline[linewidth=0.8pt,linecolor=gray](17,7)(17,8)
   \psline[linewidth=0.8pt,linecolor=gray](18,7)(18,8)
   \psline[linewidth=0.8pt,linecolor=gray](12,1)(13,1)
   \psline[linewidth=0.8pt,linecolor=gray](12,2)(13,2)
   \psline[linewidth=0.8pt,linecolor=gray](12,3)(13,3)
   \psline[linewidth=0.8pt,linecolor=gray](12,4)(13,4)
   \psline[linewidth=0.8pt,linecolor=gray](12,5)(13,5)
   \psline[linewidth=0.8pt,linecolor=gray](12,6)(13,6)
   \psline[linewidth=0.8pt,linecolor=gray](12,7)(13,7)
   \psline[linewidth=0.8pt,linecolor=gray](18,1)(19,1)
   \psline[linewidth=0.8pt,linecolor=gray](18,2)(19,2)
   \psline[linewidth=0.8pt,linecolor=gray](18,3)(19,3)
   \psline[linewidth=0.8pt,linecolor=gray](18,4)(19,4)
   \psline[linewidth=0.8pt,linecolor=gray](18,5)(19,5)
   \psline[linewidth=0.8pt,linecolor=gray](18,6)(19,6)
   \psline[linewidth=0.8pt,linecolor=gray](18,7)(19,7)
   \psline[linewidth=0.8pt,linecolor=gray](18,1.25)(19,1.25)
   \psline[linewidth=0.8pt,linecolor=gray](18,1.5)(19,1.5)
   \psline[linewidth=0.8pt,linecolor=gray](18,1.75)(19,1.75)
   \psline[linewidth=0.8pt,linecolor=gray](17.25,0)(17.25,1)
   \psline[linewidth=0.8pt,linecolor=gray](17.5,0)(17.5,1)
   \psline[linewidth=0.8pt,linecolor=gray](17.75,0)(17.75,1)
   \psline[linewidth=0.8pt,linecolor=gray](18,2.5)(19,2.5)
   \rput(8.5,-0.5){$\mathscr{T}_{3,3}$}
   \rput(15.5,-0.5){$\mathscr{T}^\varepsilon_{3,3}$}
   \end{pspicture}
 \caption{A hierarchical T-mesh $\mathscr{T}_{3,3}$ and its extension.}\label{fig:extensionT33}
\end{figure}
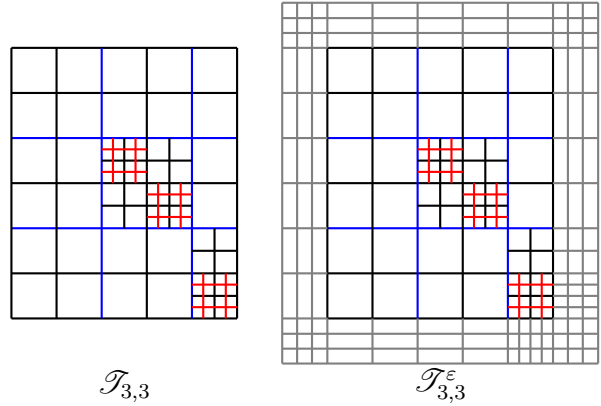
There are three isolated subdomains (one at level $0$ and two at level $1$) in
the middle of $\mathscr{T}_{3,3}^\varepsilon$.  Thus $\delta=3$.
It is easy to count that $V^+=166$, $E_H=21$ and $E_V=19$. So
\begin{align*}
  &\dim\overline{\mathbf{S}}(3,2,2,2,\mathscr{T}_{3,3}^\varepsilon)\\
  &=V^+-(3-1)E_H-(3-1)E_V+\delta+(3-1)(3-1)\\
    &=93.
\end{align*}

We can also get the dimension by counting the number of basis functions.

There are $(12-4)(13-4)=72$ basis functions over the
$11\times12$ tensor product mesh $\mathscr{T}_{3,3}^\varepsilon$. At level $1$,
$7$ basis functions are added to the basis functions set by subdividing two subdomains.
At level $2$, $14$ basis functions are added to the basis functions set by
subdividing three subdomains. Totally, there are $72+7+14=93$ basis functions,
and therefore the dimension of
$\overline{\mathbf{S}}(3,2,2,2,\mathscr{T}_{3,3}^\varepsilon)$ is
$93$.
 \end{example}

\section{Conclusions and future work}

 This paper presents a dimension formula for the spline space $\mathbf{S}(m,n,m-1,n-1,\mathscr{T})$ over certain type of hierarchical T-meshes.
 By using the smoothing co-factor method, we set up a bijection between the spline space
 $\mathbf{S}(m,n,m-1,n-1,\mathscr{T})$ and a conformality vector space $W[\mathscr{T}]$.
 Then we decompose $W[\mathscr{T}]$ into direct sum of simple linear spaces by using homological technique,
  and a dimension formula is thus obtained. At a by-product, we
 also obtain a set of basis functions for the spline space $\mathbf{S}(m,n,m-1,n-1,\mathscr{T})$.

 It is important to construct a set of proper basis functions which have nice properties for the spline space in applications.
  We will study this topic in the future.

 \section{Appendix}
 In this section, the proof of Lemma \ref{lemma:surjection} is
 given.
 \begin{proof}

  We first outline the proof strategy. For
   getting the result, we choose a set of bases
   $\{\mathbf{k}_j^i\}$ for $W[E_i]$ and construct a function $B_j^i(x,y)\in\overline{\mathbf{S}}(m,n,m-1,n-1,\mathscr{T}_i)(\cong
   W[\mathscr{T}_i])$ such that the conformality vector of $B_j^i(x,y)$
   corresponding to the vertices of $E_i$ is
   $\mathbf{k}_j^i$. Then $\pi:\,\, W[\mathscr{T}_i]\longrightarrow
   W[E_i]$ is surjective.

%

    Before we discuss the construction of $B_j^i(x,y)$, some preparations are given as follows.

   \noindent \textbf{Preparations}
    \begin{enumerate}
    \item Choose a set of bases for $W[E_i]$. If $E_i$ is horizontal,
    let $\Delta_x: -\infty<x_1<x_2<\cdots<x_{r_i}<\infty$ be the
    $x$-coordinates of the vertices of $E_i$ on $\mathscr{T}_i$.
    The B-spline basis functions
    $\{N[x_j,x_{j+1},\cdots,x_{j+m+1}]\}_{j=1}^{r_i-m-1}$ is a set
    of bases of $\overline{\mathbf{S}}(m,m-1,\Delta_x)$ and
    $\{\mathbf{k}_j^i\}_{j=1}^{r_i-m-1}$ is a set of bases of
    $W[E_i]$, where $\mathbf{k}_{j}^i$ is the conformality vector of
    $N[x_j,x_{j+1},\cdots,x_{j+m+1}]$. For a vertical interior
    l-edge, we can choose a set of bases of $W[E_i]$  similarly.

\item Define an operator on T-meshes.
  The operator $\mathbf{D}$ on a T-mesh $\mathscr{T}$ associated with an l-edge $E_i$ is defined as follows.
  $\mathbf{D}_{E_i}(\mathscr{T})=(E_i,\mathscr{T}_1)$, where
  $\mathscr{T}_1$ is the T-mesh by removing $E_i$ from $\mathscr{T}$. Then
  $\pi_1\circ\mathbf{D}_{E_i}(\mathscr{T}) = E_i$ , and
  $\pi_2\circ\mathbf{D}_{E_i}(\mathscr{T})=\mathscr{T}_1$.

    \item Define an order of
    $\mathbf{E}$. The interior l-edges $\mathbf{E}$ of
    $\mathscr{T}_{m,n}^\varepsilon$ can be divided into disjoint sets ${\mathbf E}^i$, $i=0,1,\ldots,l$,
    i.e., ${\mathbf E}=\cup_{i=0}^l {\mathbf E}^i$ with ${\mathbf E}^i\cap {\mathbf E}^j=\emptyset,\,i\not=j$.
     ${\mathbf E}^i$ is the set of interior l-edges of $\mathscr{T}_{m,n}^\varepsilon$ defined at level $i$, $i=0,1,\ldots,l$,
     where $l$ is the level of $\mathscr{T}_{m,n}$. The order
     ``$\prec$" between $\mathbf{E}^p$ and $\mathbf{E}^q$ is defined as: $\forall e^p\in \mathbf{E}^p, \forall e^q\in
  \mathbf{E}^q$, if $q<p$, then $e^q\prec e^p$. In the following, the
  order ``$\prec$" inside $\mathbf{E}^p$ is presented. Consider
  $\mathbf{E}^p$, and suppose the T-mesh $\widetilde{\mathscr{T}}^p$ is obtained by
  deleting the l-edges $\cup_{i=p+1}^{l}\mathbf{E}^i$ from
  $\mathscr{T}_{m,n}^\varepsilon$.
  Without loss of generality, it is enough to consider the case of $p=l$.

  Then, the order inside $\mathbf{E}^l$ is defined as follows with respect to the case of $l=0$ and $l>0$,
  respectively.
  \begin{itemize}
    \item $l=0$. The T-mesh $\mathscr{T}$ is a tensor product mesh
    $\mathscr{T}_\otimes$. The progress of choosing interior l-edges
    sequence is presented in the following.

    \begin{algorithmic}
  \STATE $\mathscr{T}\leftarrow\mathscr{T}_\otimes$
  \STATE $S_E\leftarrow \text{the set of interior l-edges of~} \mathscr{T}$
  \STATE $i\leftarrow 1$
  \WHILE{$S_E\neq\emptyset$}
  \IF{there is no trivial interior l-edge}
  \STATE choose any interior l-edge $E_i^0$
  \ELSE
  \STATE choose a trivial interior l-edge $E_i^0$ (i.e $W[E_i^0]=0$)
  \ENDIF
  \STATE $\mathscr{T}\leftarrow\pi_2\circ\mathbf{D}_{E_i^0}(\mathscr{T})$
  \STATE $S_E\leftarrow \text{the set of interior l-edges of~} \mathscr{T}$
  \STATE $i\leftarrow i+1$
  \ENDWHILE
\end{algorithmic}

    By this progress,
    we get a sequence of $\mathbf{E}^0$ denoted as
    $E^0_1,E^0_2,\cdots$, where $E^0_i$ is
    chosen earlier than $E^0_{i+1}$. The order is defined
    as $E^0_{i+1}\prec E^0_i$

    \item $l>0$. Select an l-edge $E\in\mathbf{E}^l$. Then $E$ must intersect with an
    $(m,n)$-subdomain $D$ at level $(l-1)$. If $E$ is vertical, we
    consider all the vertical l-edges in $\mathbf{E}^l$ which
    intersect with $D$ and arrange them from left to right. Then
    $E$ is the $i$-th l-edge. By the structure of
    $\mathscr{T}_{m,n}^\varepsilon$, $E$ is also the $i$-th l-edge in
    another $(m,n)$-subdomain at level $(l-1)$ which intersects with
    $E$. Hence we can denote the position of $E$ as $L(E)=i$.
    Considering the local tensor product structure on $D$ at level $(l-1)$,
    $L(E)=i<m$. If $E$ is horizontal, $L(E)$ can be defined in a
    similar way by arranging horizontal l-edges in $\mathbf{E}^l$ from
    top to bottom and $L(E)<n$.

    Define l-edge sets $A_1$, $A_2$, $A_3$, $A_4$, $A_5$ as
    \begin{align*}
    A_1=& \{E\in\mathbf{E}^l: \text{ If $E$ is vertical, } L(E)<m-2;\\&\text{ If $E$ is horizontal, }
    L(E)<n-2\},\\
    A_2=&\{E\in\mathbf{E}^l: \text{$E$ is horizontal, }
    L(E)=n-2\}.\\
    A_3=&\{E\in\mathbf {E}^l: \text{$E$ is vertical, } L(E)=m-2\},\\
    A_4=&\{E\in\mathbf{E}^l: \text{$E$ is horizontal, } L(E)=n-1\},\\
    A_5=&\{E\in\mathbf{E}^l: \text{$E$ is vertical, } L(E)=m-1\}.
    \end{align*}
    The progress of choosing interior l-edge sequence is presented in the following.
    \begin{algorithmic}
      \STATE $\mathscr{T}^0\leftarrow\mathscr{T}_{m,n}^\varepsilon$
      \FOR{$i=1$ to $5$}
      \STATE $\widetilde{\mathscr{T}}\leftarrow\mathscr{T}^{i-1}$
      \STATE $S_E\leftarrow \text{the set of interior l-edges of }\widetilde{\mathscr{T}}~\text{in}~ A_i $
      \STATE $j\leftarrow1$
      \WHILE{$S_E\neq\emptyset$}
      \IF{there is no trivial l-edge in $S_E$}
      \STATE choose any interior l-edge $E_{i,j}^l$ from $S_E$
      \ELSE
      \STATE choose any trivial l-edge $E_{i,j}^l$ from $S_E$
      \ENDIF
      \STATE $\widetilde{\mathscr{T}}\leftarrow\pi_2\circ\mathbf{D}_{E_{i,j}^l}(\widetilde{\mathscr{T}})$
      \STATE $S_E\leftarrow \text{the set of interior l-edges of }\widetilde{\mathscr{T}}~\text{in}~ A_i $
      \STATE $j\leftarrow j+1$
      \ENDWHILE
      \STATE $\mathscr{T}^i\leftarrow\widetilde{\mathscr{T}}$
      \ENDFOR
    \end{algorithmic}
     This progress generates an interior l-edge sequence $E_{1,1}^l$, $E_{1,2}^l$,$\cdots$, $E_{1,j_1}^l$, $E_{2,1}^l$,
    $E_{2,2}^l$, $\cdots$, $E_{2,j_2}^l$, $\cdots$, $E_{5,j_5}^l$.
    If one l-edge $e_1$ is chosen earlier than another l-edge $e_2$ in the sequence, we define $e_2 \prec e_1$.

    For an example, see Figure \ref{fig:edge-seq}, we consider the case of $l=2$. Here
      $m=4,n=3,p=6,q=6$. An l-edge in the T-mesh  belongs to the set in the right side of
      Figure \ref{fig:edge-seq} whose name has the same color as the
      l-edge. And its order of an interior l-edge sequence of $\mathbf{E}^2$ is labeled near itself.
       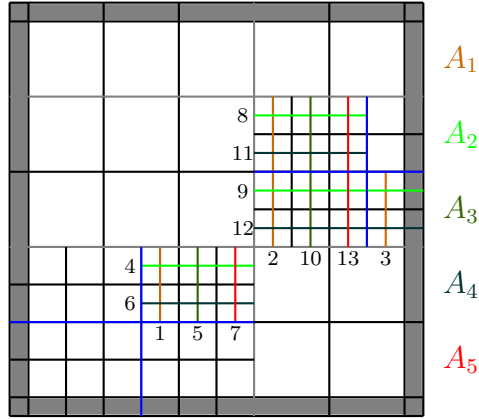
\begin{figure}[htpb]
 \begin{center}
\begin{pspicture}(0.5,0.5)(7.25,6.5)\psset{unit=1cm}
 \psline[linewidth=6pt,linecolor=gray](0.875,0.75)(0.875,6.25)
 \psline[linewidth=6pt,linecolor=gray](6.125,0.75)(6.125,6.25)
 \psline[linewidth=6pt,linecolor=gray](0.75,0.875)(6.25,0.875)
 \psline[linewidth=6pt,linecolor=gray](0.75,6.125)(6.25,6.125)
 \psline[linewidth=0.8pt](0.75,1)(6.25,1)
 \psline[linewidth=0.8pt](0.75,6)(6.25,6)
 \psline[linewidth=0.8pt](1,0.75)(1,6.25)
 \psline[linewidth=0.8pt](6,0.75)(6,6.25)
 \psline[linewidth=0.8pt](0.75,0.75)(6.25,0.75)
 \psline[linewidth=0.8pt](0.75,6.25)(6.25,6.25)
 \psline[linewidth=0.8pt](0.75,0.75)(0.75,6.25)
 \psline[linewidth=0.8pt](6.25,0.75)(6.25,6.25)
 \psline[linewidth=0.8pt](2,0.75)(2,6.25)
 \psline[linewidth=0.8pt](3,0.75)(3,6.25)
 \psline[linewidth=0.8pt,linecolor=gray](4,0.75)(4,6.25)
 \psline[linewidth=0.8pt](5,0.75)(5,6.25)
  \psline[linewidth=0.8pt](0.75,2)(6.25,2)
 \psline[linewidth=0.8pt,linecolor=gray](0.75,3)(6.25,3)
 \psline[linewidth=0.8pt](0.75,4)(6.25,4)
 \psline[linewidth=0.8pt,linecolor=gray](0.75,5)(6.25,5)
\psline[linewidth=0.8pt](1.5,0.75)(1.5,3)
\psline[linewidth=0.8pt,linecolor=gray](2.5,0.75)(2.5,3)
\psline[linewidth=0.8pt](3.5,0.75)(3.5,3)
\psline[linewidth=0.8pt](0.75,1.5)(4,1.5)
\psline[linewidth=0.8pt](0.75,2.5)(4,2.5)
\psline[linewidth=0.8pt](4.5,3)(4.5,5)
\psline[linewidth=0.8pt,linecolor=gray](5.5,3)(5.5,5)
\psline[linewidth=0.8pt](4,3.5)(6.25,3.5)
\psline[linewidth=0.8pt,linecolor=gray](4,4)(6.25,4)
\psline[linewidth=0.8pt](4,4.5)(6.25,4.5)
\definecolor{A1}{rgb}{0.8,0.4,0}
\psline[linewidth=0.8pt,linecolor=A1](2.75,2)(2.75,3)
\psline[linewidth=0.8pt,linecolor=A1](4.25,3)(4.25,5)
\psline[linewidth=0.8pt,linecolor=A1](5.75,3)(5.75,4)
\definecolor{A2}{rgb}{0,1,0}
\psline[linewidth=0.8pt,linecolor=A2](2.5,2.75)(4,2.75)
\psline[linewidth=0.8pt,linecolor=A2](4,4.75)(5.5,4.75)
\psline[linewidth=0.8pt,linecolor=A2](4,3.75)(6.25,3.75)
\definecolor{A3}{rgb}{0.2,0.4,0}
\psline[linewidth=0.8pt,linecolor=A3](4.75,3)(4.75,5)
\psline[linewidth=0.8pt,linecolor=A3](3.25,2)(3.25,3)
\definecolor{A4}{rgb}{0,0.2,0.2}
\psline[linewidth=0.8pt,linecolor=A4](4,3.25)(6.25,3.25)
\psline[linewidth=0.8pt,linecolor=A4](4,4.25)(5.5,4.25)
\psline[linewidth=0.8pt,linecolor=A4](2.5,2.25)(4,2.25)
\definecolor{A5}{rgb}{1,0,0}
\psline[linewidth=0.8pt,linecolor=A5](5.25,3)(5.25,5)
\psline[linewidth=0.8pt,linecolor=A5](3.75,2)(3.75,3)
\psline[linewidth=0.8pt,linecolor=blue](0.75,2)(4,2)
\psline[linewidth=0.8pt,linecolor=blue](4,4)(6.25,4)
\psline[linewidth=0.8pt,linecolor=blue](5.5,3)(5.5,5)
\psline[linewidth=0.8pt,linecolor=blue](2.5,0.75)(2.5,3)
\rput(6.75,5.5){\color{A1}$A_1$} \rput(6.75,4.5){\color{A2}$A_2$}
\rput(6.75,3.5){\color{A3}$A_3$} \rput(6.75,2.5){\color{A4}$A_4$}
\rput(6.75,1.5){\color{A5}$A_5$} \rput(2.75,1.85){\scriptsize{$1$}}
\rput(3.25,1.85){\scriptsize{$5$}}\rput(3.75,1.85){\scriptsize{$7$}}
\rput(2.35,2.25){\scriptsize{$6$}}\rput(2.35,2.75){\scriptsize{$4$}}
\rput(4.25,2.85){\scriptsize{$2$}}\rput(4.75,2.85){\scriptsize{$10$}}
\rput(5.25,2.85){\scriptsize{$13$}}\rput(5.75,2.85){\scriptsize{$3$}}
\rput(3.85,3.25){\scriptsize{$12$}}\rput(3.85,3.75){\scriptsize{$9$}}
\rput(3.85,4.25){\scriptsize{$11$}}\rput(3.85,4.75){\scriptsize{$8$}}
\end{pspicture}
 \caption{An example of $\mathbf{E}^2$, where
 $m=4$, $n=3$, $p=6$, $q=6$.}\label{fig:edge-seq}
 \end{center}
\end{figure}
  \end{itemize}
   \end{enumerate}

\noindent \textbf{$B_j^i(x,y)$ construction}

    Now we are ready to construct $B_j^i(x,y)$.

When $l>0$, the B-spline function $N_j^i(x,y)$ can be constructed by
the tensor product mesh over $(m,n)$-subdomains at level $(l-1)$
contained $e$. Here $N_j^i(x,y)$ is related to the knots of an
interior l-edge $e$ in the sequence associated with $A_1$(or
$A_2,A_3$). It is a tensor-product B-spline function, which equals to
$B_j^i(x,y)$ up to a nonzero constant. For trivial interior l-edges,
no function is introduced. Suppose that we construct the functions
associated with all the l-edges whose orders are ``larger" than the
order of $E$ in the sequence associated with $A_4$. We now construct
the functions associated with $E$. $E$ is a non-trivial horizontal
l-edges and we denote the $x$-coordinates of $E$ as
$\{x_1,x_2,\cdots,x_d\}$. Let $\alpha$ be the number of vertices
formed by two l-edges with level $l$ intersection. This type of
vertex is called  an $(l,l)$-vertex later. The $i$-th knots sequence
of $E$'s $x$-coordinates is $\{x_i,\cdots,x_{i+(m+1)}\}$. According
to the structure of the current mesh, $\alpha\in\{0,1,2\}$.
\begin{enumerate}
\item $\alpha=0$. If $l\geq2$, a B-spline can be constructed
over a tensor product mesh over all the $(m,n)$-subdomains at level
$(l-2)$ containing $E$'s $i$-th knot sequence. If $l=1$, we can
construct a B-spline according to the tensor product mesh of level $
0$. The required function $B_j^i(x,y)$ is the same as the B-spline
function up to a nonzero constant.

\item $\alpha=1$. If $l\geq 2$, consider the unique $(l, l)-$vertex, denoted as $Q$.
Assume the $(m,n)$-subdomain which contains $Q$ at level $(l-2)$ is
$\Sigma$. If the interior vertical l-edge that contains $Q$ goes
through $\Sigma$, we can construct a B-spline similar to the case of
$\alpha=0$. If it does not go through $\Sigma$, one of its endpoint
must be within $\Sigma$ and another is outside of $\Sigma$ according
to the structure of the special hierarchical T-mesh. Suppose the
intersection point between the vertical l-edge and $\Sigma$ is $R$.
We extend the l-edge from the endpoint within $\Sigma$ and cut the
boundary of $\Sigma$ at $P$. Then the $n+2$ vertices can be used to
construct a B-spline $N_1(x,y)$ by a tensor product mesh over all
$(m,n)$-subdomains containing the given $(m+2)$ vertices. These
$n+2$ vertices consist of $P$, $R$, and $n$ vertices in the original
vertical l-edge within $\Sigma$. Therefore, the conformality vector
of $N_1(x,y)$ is nontrivial at these $(m+2)$ vertices and  trivial
at the new vertices without $P$ in the process of extending. Suppose
$k_1$ is the conformality factor at $P$. If another function
$g(x,y)$ is constructed such that its conformality factor at $P$ is
$-k_1$ and the conformality factors at the new vertices without $P$
and the given $(m+2)$ vertices are zeros, then
$f(x,y)=N_1(x,y)+g(x,y)$ differs by a multiplier $\lambda$ ($\neq 0$)
with the function $B_i^j(x,y)$ which we need by the
linear property of $\mathscr{K}$ and Example \ref{example:Bspline
surface}.

Now we discuss the construction of $g(x,y)$. If we remove $Q$
from the set consisting of $P$ and all the vertices in the original
vertical l-edge, the number of remaining vertices is not less than
$(n+2)$. We choose $(n+2)-$vertices from these remaining vertices
and $P$ must be one of them. Similar to the case of $\alpha=0$,
there is no $(l,l)$-vertex in these $(n+2)-$vertices. So, a
B-spline $N_2(x,y)$ is constructed and we denote $k_2$ as the
conformality factor of $N_2(x,y)$ at $P$. By Example
\ref{example:Bspline surface}, $k_2\neq 0$ and the conformality
vector of $N_2(x,y)$ is trivial at the new vertices without $P$ and
those given $(m+2)$ vertices according to $Q$ removed. Then we define
$g(x,y)=-k_1/k_2N_2(x,y)$.

For an example, see Figure \ref{fig:construction}, where $m=3$, $n=3$. Here
$P$ is shown as ``$\Box$" in the figure. $N_1(x,y)$ is
constructed over the tensor product mesh shown in the upper right of
this figure and the one in the lower right is used to construct
$N_2(x,y)$. Compute $k_1$ of $N_1(x,y)$ and $k_2$ of $N_2(x,y)$ at
$P$ as we have shown in Examples \ref{example:Bspline curve } and
\ref{example:Bspline surface}. Then  $f(x,y)$ is given by
 \begin{align}\label{eq:combination}
    f(x,y)&=N_1(x,y)-\frac{k_1}{k_2}N_2(x,y).
 \end{align}
\begin{figure}[htpb]
\begin{center}
\begin{pspicture}(0,0)(11,9)\psset{unit=1cm}
   \psgrid[subgriddiv=1,gridlabels=0pt](0,4)(2,6)
   \psgrid[subgriddiv=2,gridlabels=0pt](0,2)(2,4)
   \psgrid[subgriddiv=2,gridlabels=0pt](2,2)(4,3)
   \psgrid[subgriddiv=2,gridlabels=0pt](3,4)(4,4)
   \psgrid[subgriddiv=2,gridlabels=0pt](2,5)(3,6)
   \psgrid[subgriddiv=2,gridlabels=0pt](3,3)(4,4)
   \psgrid[subgriddiv=4,gridlabels=0pt](2,3)(3,5)
   \psgrid[subgriddiv=4,gridlabels=0pt](3,4)(4,6)
   \psline[linewidth=0.8pt,linecolor=red](2.5,4.625)(3.25,4.625)
   \psline[linewidth=0.8pt,linecolor=red](3.125,4.5)(3.125,5.5)
\psgrid[subgriddiv=1,gridlabels=0pt](7,7)(9,9)
   \psgrid[subgriddiv=2,gridlabels=0pt](7,5)(9,7)
   \psgrid[subgriddiv=2,gridlabels=0pt](9,5)(11,6)
   \psgrid[subgriddiv=2,gridlabels=0pt](10,7)(11,7)
   \psgrid[subgriddiv=2,gridlabels=0pt](9,8)(10,9)
   \psgrid[subgriddiv=2,gridlabels=0pt](10,6)(11,7)
   \psgrid[subgriddiv=4,gridlabels=0pt](9,6)(10,8)
   \psgrid[subgriddiv=4,gridlabels=0pt](10,7)(11,9)
   \psline[linewidth=0.4pt,linecolor=red](9.5,7.625)(10.25,7.625)
   \psline[linewidth=0.4pt,linecolor=red](10.125,7.5)(10.125,8.5)
   \psline[linewidth=0.4pt,linecolor=red,linestyle=dashed](10.125,8.5)(10.125,7)
   \psdots[dotstyle=square](10.125,7)
   \psdots[dotstyle=triangle](10.125,7.5)(10.125,7.625)(10.125,7.75)(10.125,8)(9.5,7.625)(9.75,7.625)(10,7.625)(10.25,7.625)
  \psgrid[subgriddiv=1,gridlabels=0pt](7,2)(9,4)
   \psgrid[subgriddiv=2,gridlabels=0pt](7,0)(9,2)
   \psgrid[subgriddiv=2,gridlabels=0pt](9,0)(11,1)
   \psgrid[subgriddiv=2,gridlabels=0pt](10,2)(11,2)
   \psgrid[subgriddiv=2,gridlabels=0pt](9,3)(10,4)
   \psgrid[subgriddiv=2,gridlabels=0pt](10,1)(11,2)
   \psgrid[subgriddiv=4,gridlabels=0pt](9,1)(10,3)
   \psgrid[subgriddiv=4,gridlabels=0pt](10,2)(11,4)
    \psdots[dotstyle=square](10.125,2)
    \psdots[dotstyle=triangle](10,3.5)(10.125,3.5)(10.25,3.5)(10.5,3.5)(10.75,3.5)(10.125,2.5)(10.125,2.75)(10.125,3)(10.125,3.5)
       \psline[linewidth=0.4pt,linecolor=red](9.5,2.625)(10.25,2.625)
   \psline[linewidth=0.4pt,linecolor=red](10.125,2.5)(10.125,3.5)
   \psline[linewidth=0.4pt,linecolor=red,linestyle=dashed](10.125,3.5)(10.125,2)
   \psline[linewidth=10pt,arrows=->,linecolor=lightgray](4.5,4)(6.5,4)
 \end{pspicture}
  \caption{Construction of the linear combination of B-splines in the case of $m=n=3$.}\label{fig:construction}
  \end{center}
\end{figure}
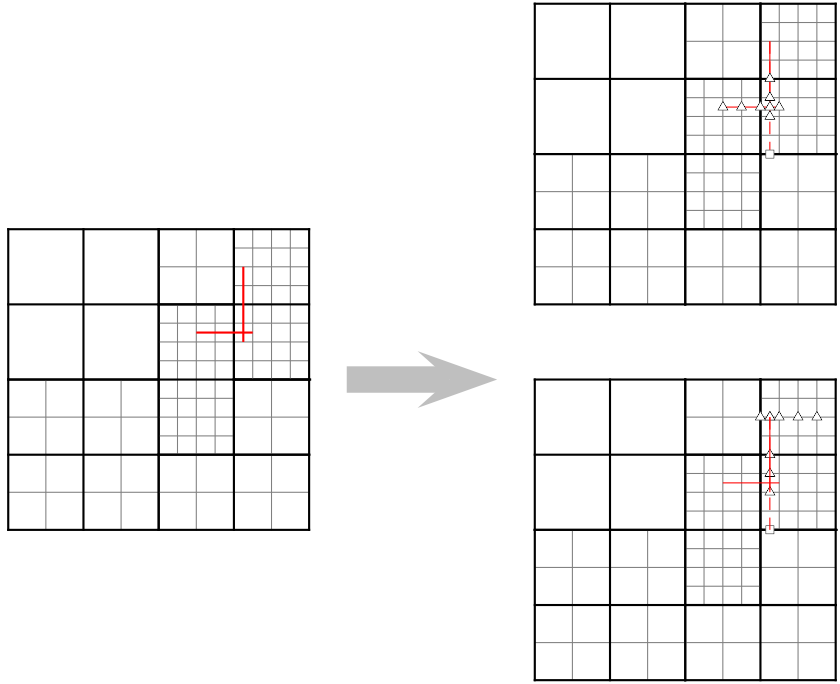
 When $l=1$, a B-spline associated with the given $m+2$
 vertices can be constructed by the tensor product mesh of level 0.

\item $\alpha=2$. If $l\geq2$, there are two $(l,l)-$vertices $P_1$ and $P_2$ in the given $(m+2)$ vertices.
One of the two vertical l-edges at $P_1, P_2$ respectively must go
through $\Sigma'$ that occupied by $(m,n)$-subdomains containing $P_1$
and $P_2$ at level $(l-2)$. So, we can construct the function by the
method presented in the case of
$\alpha=1$. When $l=1$, the function can be constructed by the
tensor product mesh of level 0 and the case $l\geq 2$.
\end{enumerate}

When all the l-edges in the sequence associated with $A_4$  are
deleted, there are no $(l, l)-$vertices at the l-edges in $A_5$. So
we can construct the function we need by the method presented in the
case of $\alpha=0$.

When $l=0$, by the order defined on $\mathbf{E}^0$ and the tensor
product mesh at level 0, B-spline surfaces can always be
constructed. So as stated previously, $\pi$ is surjective.

   \end{proof}
\section*{Acknowledgement}
The authors are supported by 973 Program 2011CB302400, the NSF of China (No. 60873109, 11031007 and 61073108), Program for New Century Excellent Talents in University (No. NCET-08-0514).

\bibliographystyle{ieeetr}
\bibliography{ref}

\end{document}